\documentclass[reqno]{amsart}
\usepackage{amsmath,amssymb}
\usepackage{graphics,epsfig,color}
\usepackage[mathscr]{euscript}
\usepackage{mathtools}
\usepackage{enumerate} 
\usepackage{cancel}
\newtheorem{theorem}{Theorem}[section]
\newtheorem{proposition}[theorem]{Proposition}

\theoremstyle{definition}
\newcounter{assumption}
\newcounter{definition}
\newtheorem{definition}[definition]{Definition}
\newtheorem{assumption}[assumption]{Assumption}
\newcounter{remark}
\newtheorem{remark}[remark]{Remark}
\numberwithin{assumption}{section}
\numberwithin{definition}{section}
\numberwithin{remark}{section}
\numberwithin{equation}{section}
\numberwithin{theorem}{section}
\newcommand{\mc}[1]{{\mathcal #1}}
\newcommand{\bb}[1]{{\mathbb #1}}
\newcommand{\ms}[1]{{\mathscr #1}}

\newcommand{\R}{\mathbb{R}}
\newcommand{\eps}{\varepsilon}
\newcommand{\vt}{\vartheta}
\newcommand{\e}{\epsilon}
\newcommand{\ee}{\mathrm{e}}
\newcommand{\id}{{1 \mskip -5mu {\rm I}}}

\newcommand{\Ent}{\mathop{\rm Ent}\nolimits}
\newcommand{\E}{\mathop{\rm E\,\!}\nolimits}
\newcommand{\HH}{\mathop{\rm H\,\!}\nolimits}

\newcommand{\de}{\mathop{}\!\mathrm{d}}
\newcommand{\opde}[1]{{\dfrac {\mathrm{d} \phantom{#1}}{\mathrm{d} {#1}}}}

\newcommand{\MM}{{\mathcal M}}
\newcommand{\SM}{{\mathcal S}}
\newcommand{\MB}{{\ms M}}
\newcommand{\SB}{{\ms S}}

\title{Entropy dissipation inequality for general binary collision models}
\author[G.\ Basile]{Giada Basile}
\address{Giada Basile \hfill\break \indent
  Dipartimento di Matematica, Universit\`a di Roma `La Sapienza'
  \hfill\break \indent
  P.le Aldo Moro 2, 00185 Roma, Italy}
\email{basile@mat.uniroma1.it}
\author[D.\ Benedetto]{Dario Benedetto}
\address{Dario Benedetto \hfill\break \indent
  Dipartimento di Matematica, Universit\`a di Roma `La Sapienza'
  \hfill\break \indent
  P.le Aldo Moro 2, 00185 Roma, Italy}
\email{benedetto@mat.uniroma1.it}

\begin{document}


\maketitle

\begin{abstract}
  We introduce a ``two-particle factorization''
  condition which allows us to formulate the homogeneous Boltzmann equation
  for non-reversible collision kernels in terms of an entropy inequality.
  This formulation yields an H-Theorem.
  We provide some examples of non-reversible binary collision models 
  with a concentration/dispersion mechanism,  as in opinion
  dynamics, which satisfy this condition.
  As a preliminary step, we also provide
  an  analogous  variational formulation of non-reversible continuous time
  Markov chains, expressed in terms of an entropy dissipation inequality.
\end{abstract}

\section{Introduction}
\label{sez:intro}

In recent years, a fruitful approach for the study of various
equations in mathematical physics has been their formulation in terms
of gradient flows, starting from the pioneering work \cite{JKO}, and
the general theory developed in \cite{AGS} for diffusion.  See also
\cite{otto, cddw} for other equations.  In \cite{Ma,Mi,Er} the
Fokker-Planck equation associated with reversible continuous time
Markov chains has been formulated in terms of energy variational
inequalities.  For linear inhomogeneous Boltzmann equations,
a gradient flow formulation
has been introduced
in \cite{BBB}, 
expressed as  an entropy dissipation inequality.

A more challenging case is the Boltzmann equation.
A gradient flow formulation in terms of a metric interpretation of the
entropy inequality 
has been shown in  \cite{Er2} in the
space-homogeneous case, see also \cite{EHe} for a generalization to a
non-homogeneous
model. 
In \cite{BBO} a slightly different formulation has been proposed
by relating
the entropy inequalities to the large deviation
rate function for the Kac's walk, which is an
underlying microscopic model of the homogeneous Boltzmann equation. {In this setting the flux is considered as a dynamical variable, and the related  large deviation rate function  is  the relative entropy of the flux respect to the typical one. 
  The variational structure naturally follows  from a chain rule for the entropy,
  which involves also 
  the large deviation rate function of the time-reversed Kac's dynamics.
As a consequence, an $H$-theorem for the dynamics and for the time-reversed one holds.}
For the connection between entropy dissipation inequality and large deviations
see also \cite{ADPZ,MiPR}.

This paper explores the case of irreversible processes, {namely when the microscopic underlying dynamics does not satisfy the detailed balance condition. Our goal is  to formulate a variational characterization of the irreversible dynamics in terms of an entropy inequality, which still involves the large deviation rate function of the dynamics and of the time-reversed one.}

{As first step we consider non-reversible, continuous time Markov chains on a general state space $\mc X$.
  Under the assumption that the time-reversed transition rate exists with respect to an invariant measure $\pi$ of the chain and assuming some sufficient conditions on the transition rate, we  prove a chain rule for the entropy.
 More precisely,  denoting by $\Ent(\cdot\vert\cdot)$ the relative entropy between two positive measures, we obtain
  \begin{equation*}
  \Ent (P_T|\pi) + \Ent(\mc V|\Upsilon _\# \hat{\mc V}^P) =  \Ent(P_0|\pi) +\Ent(\mc V|\mc V^P), 
  \end{equation*}
  where $(P_t)_{t\geq 0}$ is a probability measure for the chain, $\mc V$ is the related flux, $\mc V^P$, $\hat{\mc V}^P$ are the typical flux of the dynamics and of the time-reversed process, respectively, while $\Upsilon$ is the operator which revert the states before and after the jump. The connection with  large deviations is the following. The right-hand side of the chain rule  is the large deviation rate function for the
  empirical measure and flux of $N$ independent copies of the chain, while the left-hand side is the large deviation rate function for empirical measure and flux of the time-reversed dynamics.
The equality coincides with the so called
generalized  Onsager-Machlup relation, see
  \cite{BDGJL}.
}

As a consequence, we establish
a variational formulation
in terms of an entropy dissipation inequality for
the probability measure $(P_t)$ and the related flux $\mc V$.
In particular,
we show that the Kolmogorov forward equation for $P$ is equivalent to the inequality
$$\Ent (P_T|\pi) + \Ent(\mc V|\Upsilon _\# \hat {\mc V}^P) \le \Ent(P_0|\pi),$$
%where $\pi$ is the equilibrium measure, $\Ent$ is the relative entropy between probability measures,
%$\E$ is the relative entropy between positive measures,
%$Q^P$ is the typical flux
%and $\Upsilon _\# \hat Q^P$ is the flux
%of the time-reversed process.
It turns out that this inequality
is fulfilled if and only if $\mc V={\mc V}^P$, i.e. if $P$  satisfies 
the Kolmogorov equation.

%In \cite{BBO} an analogous formulation has been provided
%for the homogeneous Boltzmann equation.
%The underlying microscopic dynamics, i.e. the Kac's walk,
%is reversible, in the sense that the detailed balance condition holds
%(see also \cite{BBBO} for a Kac-type walk).
{As second and more ambitious step we aim to give a variational inequality for a Boltzmann equation
when the underlying microscopic dynamics does not satisfy the detailed
balance condition, in the same spirit of  \cite{BBO}, i.e. by deriving
the chain rule for the entropy involving
the large deviation rate function of the microscopic dynamics
and of the time reversed one.
We stress that in the non-reversible case, the kinetic limit of the microscopic
time-reversed
dynamics is not necessarily described by a Boltzmann type equation.
Here we prove the chain rule  assuming a 
``two-particle factorization'' condition on the collision kernel (see \eqref{eq:cond}).
This condition allows us to associate to the Boltzmann equation 
a time-reversed dynamics which is described 
by a Boltzmann equation too.
From the microscopical point of view, the two-particle factorization condition is related to the jump Markov process for particle pairs,
in particular it is equivalent to the fact that 
the two-particle equilibrium measure factorizes.

}

Starting from the chain rule, we derive the variational formulation
for the homogeneous Boltzmann equation with non-reversible collision kernel, namely 
$$\Ent(P_T|\pi) + \E(Q|\Upsilon_\# \hat Q^P) \le  \Ent(P_0|\pi),$$
which in particular implies the H-theorem.

A variational formulation in more general case is a challenging problem. A further step
in this direction will be the subject of future work (see \cite{BBC}).

The prototype of Boltzmann equation with non-reversible collision kernels
is the kinetic equation for a granular gas.
Other cases of non-reversible collision kernel
have been introduced in effective models
with binary interactions, such as opinion dynamics and market economy models
\cite{toscani2}.
For these models,
the stationary measure is either singular
or unknown, and in both cases, the H-Theorem is not available.
We construct some examples of non-reversible Boltzmann equation
which have a variational formulation and for which the H-Theorem
characterizes the equilibrium.

The paper is organized as follows.
In Section 2 we present variational formulations for
{non-reversible continuous-time Markov chains.
An earlier description of the result is in \cite{magnanti}.}
As a by product we obtain the result for reversible Markov chain obtained in \cite{BBB}.

In Section 3 we consider the non-linear case. We start by presenting the
variational formulation of the homogeneous Boltzmann equation
obtained in \cite{BBBH}, then we consider the case of non-reversible
collision kernels, and we formulate the two-particle factorization
condition.

In Section 4
we consider some examples. In particular, we introduce 
one-dimensional
binary  processes with a concentration/dispersion mechanism (as in
opinion dynamics \cite{toscani}), which may be of interest
for the applications.

\section{Non-reversible jump processes}
\label{sez:markov}
In this section we discuss the entropy dissipation inequality for
reversible and non-reversible continuous time Markov chains.

Assume $\mc X$ be a Polish space, i.e.  a metrizable, complete
and separable topological space.
We denote by
$\mathcal P (\mc X)$ the set of probabilities on
$(\mc X, \mathcal B)$, where $\mathcal B$ is the Borel
$\sigma-$algebra, which we endow with the topology of weak convergence.
We consider a jump process with transition kernel
$r(x,\de y)$, i.e. such that
for any $x\in\mc X$, $r(x,\cdot)$ is a measure on $\mc X$ with finite total
mass,
and for any $B\in\mathcal B$, the map $\mc X \ni x\to r(x,B)\in [0,+\infty)$ is measurable.

We assume that there exists a 
stationary probability measure
$\pi \in \mc P(\mc X)$ for the chain,
i.e.
\begin{equation}
  \label{eq:eqmarkov}
\pi(\de x) \int_{y\in \mc X} r(x,\de y)  =
\int_{y \in \mc X} \pi(\de y) r(y, \de x),
\end{equation}
which is strictly  positive on any open subset of $\mc X$.
We call $\sigma:\mc X \times \mc X \to [0,+\infty)$
the  Borel function such that
\begin{equation}
  \label{eq:sigma}
  r(x,\de y) = \sigma(x,y) \pi (\de y).
\end{equation}
Given $T>0$, $P\in C([0,T],\mc P(\mc X))$
is a weak solution to the Kolmogorov forward equation 
for the Markov
chain with transition kernel $r$ 
 if for any 
 $\phi \in C_0([0,T]\times \mc X)$, with continuous derivative with
 respect to $t$,
 \begin{equation} \label{eq:kolm}
  P_T(\phi_T)-P_0(\phi_0)-\int_0^T \de t
  P_t(\partial_t \phi_t)= \int_0^T \de t \int_{\mc X^2} r(x,\de y)
  P_t(\de x) (\phi_t(y) - \phi_t(x)).
\end{equation}
Let  $\MM$ be is the set of positive finite measures on
$[0,T]\times \mc X \times \mc X$, and 
$\SM$  the subset of $C([0,T], \mc P(\mc X))\times
\MM$ of the pairs $(P,\mc V)$ such that for any $\phi \in
C_0([0,T],\mc X)$, with continuous derivative with respect to $t$,
\begin{equation}
  \label{eq:cont} P_T(\phi_T)-P_0(\phi_0)-\int_0^T \de t
  P_t(\partial_t \phi_t)=\mc V(\overline \nabla^2 \phi)
\end{equation} where $\overline \nabla^2 \phi(x,
y)\coloneqq \phi(x)-\phi(y)$.
Equation \eqref{eq:cont}  is the conservation law of the probability  
$P\in C([0,T], \mc P(\mc X))$
for a process with jumps distributed according to the measure $\mc V\in \MM$,
which we call ``flux''.

\begin{definition}[Measure-flux solutions to the Kolmogorov equation]
  \label{def:markov}

  Fix $T>0$. We say that a measure-flux pair $(P,\mc V)\in \SM$
  is a              solution to the Kolmogorov equation
  if and only if $\mc V = \mc V^P$, where
  $$\mc V^P(\de t,\de x, \de y) \coloneqq \de t \, r(x,\de y) P_t(\de x).$$
\end{definition}
The above definition is justified by the fact that
if $(P,\mc V^P)\in \SM$ then
$P\in C([0,T],\mc P(\mc X))$ solves \eqref{eq:kolm}. {In this case,
  since  $\mc V^P$ is a finite measure
  Eq.s \eqref{eq:kolm} and \eqref{eq:cont} hold also for $\phi \in C_b(\mc X)$.}

Following \cite{BBB}, in which  an inhomogeneous
linear Boltzmann type 
equations is considered, 
we show that, under suitable condition on $r$, 
if $(P,\mc V)\in \SM$ the constitutive equation
$\mc V=\mc V^P$ is equivalent to an entropy dissipation inequality.

Given $\mu,\nu\in \mc P(\mc X)$,
the relative entropy $\Ent (\mu|\nu)$ is  
$$
\begin{aligned}
  \Ent(\mu|\nu)
  &\coloneqq  \sup_{\varphi \in C_b(\mc X)}
    \mu(\varphi) - \nu(\ee^{\varphi} -1 )
    =
    \sup_{\varphi \in C_b(\mc X)} \mu(\varphi) - \log \nu(\ee^{\varphi}) \\
  =&\begin{cases}
     \int \de \mu
      \log \frac {\de \mu}{\de \nu} & \text{ if } \mu \ll \nu \\
      +\infty & \text{ otherwise }
    \end{cases}
\end{aligned}
$$
We also define the relative entropy of two positive
measures $\mc V,\tilde {\mc V} \in \MM$
$$\E(\mc V|\tilde {\mc V}) \coloneqq \sup_{F\in C_b([0,T]\times \mc X \times \mc X)} (\mc V(F)-
\tilde {\mc V}(\ee^F -1))$$
which it turns out to be
$$\E(\mc V|\tilde {\mc V}) =
\begin{cases}
  \int \de {\mc V} 
  \log \frac {\de \mc V}{\de \tilde {\mc V}}
  - \de {\mc V} + \de \tilde {\mc V}
  & \text{ if } \mc V \ll \tilde {\mc V}\\
  +\infty & \text{ otherwise}
\end{cases}
$$
Note that both $\Ent $ and $\E$
are non-negative convex and lower semi-continuous functionals of their
two arguments. Moreover, $\Ent(\mu|\nu) = 0$ if and only if  $\mu=\nu$ and
$\E(\mc V|\tilde {\mc V})  = 0$ if and only if  $\mc V=\tilde {\mc V}$.

Now we state the precise assumption.
 
 \begin{assumption}
   \label{ass:markirr}

   $~$

  \begin{itemize}
  \item[(i)]

    There exists the transition rate of the time-reversed
    process $\hat r$ defined as
    \begin{equation}
      \label{eq:rate-inverso}
      r(x, \de y)   \pi(\de x) = \hat r(y, \de x) \pi (\de y),
    \end{equation}
    i.e. 
      $ \hat r(x, \de y) = \sigma(y,x)   \pi(\de y)$.    
\item[(ii)]
  $\sigma$ is  bounded.
  \end{itemize}
\end{assumption}
Observe that, by item (ii), the notion  of measure-flux solution and of \eqref{eq:kolm} are equivalent.

Denote by  $\hat {\mc V}^P$ the flux of the time-reversed process, namely 
\begin{equation}
  \label{eq:hatQ}
  {\hat {\mc V}}^P(\de t, \de x, \de y)
  \coloneqq  \de t \, \hat r(x,\de y) P_t(\de x) 
\end{equation}

Let $\Upsilon: [0,T] \times \mc X \times \mc X \to
[0,T] \times \mc X \times \mc X$ be the map that 
exchanges the incoming and the outgoing states, namely 
$\Upsilon (t,x,y)= (t,y,x)$.

\begin{proposition}[Entropy balance for non-reversible Markov chains]

  $~$
  
  Under Assumption \ref{ass:markirr}, for each  $(P,\mc V)\in \SM$, 
  \begin{equation}
    \label{eq:EBI}
    \Ent (P_T|\pi) + E(\mc V|\Upsilon_\# {\hat {\mc V}}^P) =
    \Ent (P_0|\pi) + E(\mc V|\mc V^P)
  \end{equation}
  We intend the above equation in the sense that if either side is finite,
then also the other one is finite and equality holds.
\end{proposition}

The proof is obtained by showing that \eqref{eq:EBI}
holds for a regularization of the pair $(P,\mc V)$,
and then by removing the regularization.
The details are in  Appendix \ref{app:prova}.
Here we only remark that the possibility of writing the entropy balance
$\Ent(P_T|\pi) -\Ent(P_0|\pi)$ as the difference 
$\E(\mc V|\mc V^P) - E(\mc V|\Upsilon_\# {\hat {\mc V}}^P)$
depends on the fact that
if  $P_t\ll \pi$ for all $t$, the total flux
of the process is the same as the total flux of the reversed one:
$$
\mc V_\sigma^P(1) = \hat {\mc V}_\sigma^P(1)  =
\Upsilon_\# \hat {\mc V}_\sigma^P(1),$$
as follows from the equilibrium condition \eqref{eq:eqmarkov}.

\begin{remark}
  Observe that the right-hand side of \eqref{eq:EBI} 
  is the  large deviation rate function
  for the empirical measure and flux
  constructed by taking $N$ independent copies of the chain.
  Moreover, the left-hand side
  is large deviation rate function for the
  time-reversed dynamics.
  In particular the equality coincides with the so called
  Onsager-Machlup relation, see
  \cite{BDGJL}.

\end{remark}

\begin{remark}
  If we set $\mc V=\mc V^P$ in the entropy balance equation \eqref{eq:EBI},
  we get the $H$-theorem for the chain, 
  while, if $\mc V= \Upsilon_\# \hat {\mc V}^P$, we get the
  $H$-theorem for the time-reversed process.
\end{remark}

\begin{proposition}[Variational formulation for non-reversible Markov chain]

   Let  $(P,\mc V)\in\SM$ be such that  $\Ent(P_0|\pi) < +\infty$.
   The following assertions are equivalent.

   \begin{itemize}
   \item[(i)] $P$ solves the Kolmogorov equation \eqref{eq:kolm}.

   \item[(ii)] 
   $\Ent(P_T|\pi) + \E(\mc V|\Upsilon_\# {\hat {\mc V}}^P) \le
   \Ent(P_0|\pi).$
 \end{itemize}
\end{proposition}
\begin{proof}
  By  the entropy balance \eqref{eq:EBI},
  $$\Ent(P_T|\pi) + \E(\mc V|\Upsilon_\# {\hat {\mc V}}^P) \ge
  \Ent(P_0|\pi),$$
  where the equality holds if and only if $\E(\mc V|\mc V^P)=0$, i.e. $\mc V = \mc V^P$.  
\end{proof}
Observe that, by the convexity of the relative entropy, the solution is unique
(see \cite{Gi} and \cite[Theorem 2.6]{BBB}).

\subsection*{Unbounded transition kernels}

Denote by $\lambda:\mc X \to [0,+\infty)$ and
$\hat \lambda:\mc X \to [0,+\infty)$ the jump rates
$$\lambda (x) \coloneqq \int_{\mc X} r(x,\de y) =  \int_{\mc X} \sigma(x,y) \pi(\de y).$$
$$\hat \lambda (x) \coloneqq \int_{\mc X} \hat r(x,\de y) =  \int_{\mc X} \sigma(y,x) \pi(\de y).$$
Under Assumption \ref{ass:markirr}, $\lambda$ and $\hat \lambda$ are  continuous and bounded.
In order to generalize the previous results to unbounded transition kernels,
one needs to verify in particular the condition 
\begin{equation}\label{Pl}\sup_{t\in [0,T]} \int_{\mc X} P_t(\de x) \lambda(x)\end{equation}
that assures the equivalence of \eqref{eq:kolm} and measure-flux solutions.
The validity of the above condition depends on the details of the model.

Here we propose a generalization in the same spirit of \cite{BBB}.
  In particular, we assume 
  \begin{equation}
    \label{fE}\sup_{t\in [0,T]} \Ent(P_t|\pi) < +\infty,
  \end{equation}
  and we require that
  $\lambda$ and $\hat \lambda$ have all exponential moments with respect
  to $\pi$,
  namely
  \begin{equation}\label{exp_m}
    \pi(\ee^{\gamma \lambda}) <+\infty, \qquad \pi(\ee^{\gamma \hat \lambda}) <+\infty, \quad \forall \gamma>0.
  \end{equation}
Observe that under the above assumptions, by basic entropy inequality \eqref{Pl} is satisfied.
Adapting the argument in \cite[Theorem 2.6]{BBB}, one can prove the following
statement.

\begin{proposition}
  \label{prop:3}
  Assume that $\lambda$ and $\hat \lambda$ have all exponential moments with respect to $\pi$.
  Consider  $(P,\mc V)\in \SM$ such that  
  $\sup_{t\in [0,T]} \Ent(P_t|\pi) < +\infty$.
  Then the following assertions are equivalent.
  
  \begin{itemize}
  \item[(i)] $P$ solves the Kolmogorov equation \eqref{eq:kolm}.
  \item[(ii)] 
    $\Ent(P_T|\pi) +
    \E(\mc V|{{\mc V}}^P) +
    \E(\mc V|\Upsilon_\# {\hat {\mc V}}^P) \le
    \Ent(P_0|\pi).$
  \end{itemize}
\end{proposition}
The key point is that the reverse inequality in item (ii) holds
for any $(P,\mc V)\in \SM$. Indeed, if the two functional
$\E(\mc V|{{\mc V}}^P)$ and $\E(\mc V|\Upsilon_\# {\hat {\mc V}}^P)$
are finite,  \eqref{fE} and \eqref{exp_m} imply the chain rule \eqref{eq:EBI}.

\vskip.3cm
In the reversible case, i.e. when the detailed balance condition holds,
$\sigma(x,y) = \sigma(y,x)$, $r=\hat r$, $\hat {\mc V}^P = \mc V^P$.
In this case, inequality in item (ii) in Proposition \ref{prop:3}
can be
rewritten in terms of the Dirichlet form of the square root of $f$
and a kinematic term (see \cite[Equation 2.23]{BBB}).
More precisely,
for  $P\in C([0,T]; \mc P(\mc X))$, with $P_t\ll \pi$ for any $t\in [0,T]$,
define the functional $ \mathrm D$ 
\begin{equation}
  \label{def:D}
   \mathrm D(P_t) = \int \sigma(x,y)\pi(\de x)\pi(\de y)
   (\sqrt{ f_t(x)} -   \sqrt{ f_t(y)})^2,
 \end{equation}
 where 
 $f_t =\frac {\de P_t}{\de \pi}$,
 and the flux  $R^P$
 \begin{equation}
   \label{def:RP}
   \de R^P \coloneqq \de t\, r(x,\de y)\pi(\de x) \sqrt{f_t(x) f_t(y)}
   =  \de t \,\sigma(x,y)  \pi(\de x) \pi(\de y) \sqrt{f_t(x) f_t(y)}.
 \end{equation}

 The following equality holds
 $$\E(\mc V|\mc V^P) + \E(\mc V|\Upsilon_\# \mc V^P)  = 2\E(\mc V|R^P) +
 \int_0^T \de t\,
 \mathrm D (P_t).$$
 Note that this decomposition holds also
for  non-reversible chains, and in this case
$\mathrm{D}$ depends  only on the symmetric part of $\sigma$.

\section{Binary collision models}
\label{sez:binary}

In the usual notation for the Boltzmann equation,
a binary collision model is characterized by a
transition kernel $B(v,v_*,\de v',\de v'_*)$,
where  $v,v_*\in \mc X$ are the incoming velocities (or states)
and  $v',v'_*\in \mc X$ are the outgoing ones.

Fixed $T>0$, the one particle probability measure
$P\in C([0,T],\mc X)$ satisfies the homogeneous Boltzmann equation,
whose weak form 
reads
\begin{equation}\begin{aligned}
  \label{eq:kac}
  P_T(\phi_T)&-P_0(\phi_0)-\int_0^T \de t P_t(\partial_t \phi_t)
  \\&=
  \frac 12
  \int_0^T \de t \int_{\mc X^4} B(v,v_*,\de v',\de v'_*) P_t(\de v)P_t(\de v_*)
  (\overline \nabla^4 \phi_t)(v,v_*,v',v'_*),
  \end{aligned}
\end{equation}
where $\phi$ is any function in 
$C_0([0,T],\mc X)$
with continuous derivative w.r.t. to $t\in [0,T]$,
and  
$\overline \nabla^4 \phi (v,v_*,v',v'_*)
\coloneqq \phi(v') + \phi(v'_*)
-\phi(v)-\phi(v_*)$.
Without loss of generality, we can assume
$$B(v,v_*,\de v',\de v'_*) = B(v,v_*,\de v'_*,\de v') = B(v_*,v,\de v',\de v'_*).$$

We rewrite this equation in
terms of a measure-flux pair.
We denote by $\MB$  the set of positive finite measures 
on $[0,T]\times \mc X^2\times \mc X^2$
with the symmetry
$Q(\de t; \de v, \de v_*, \de v', \de v_*' ) =
Q(\de t; \de v, \de v_* , \de v_*' , \de v_* ) =
Q(\de t; \de v_*, \de v , \de v' , \de v_*' )$.
Let 
$\SB$ be  the subset of $C([0,T], \mc P(\mc X))\times
\MB$ of the pairs $(P,Q)$ such that for any $\phi \in
C_0([0,T],\mc X)$, with continuous derivative with respect to $t$,
\begin{equation}
  \label{eq:Bcont}
  P_T(\phi_T)-P_0(\phi_0)-\int_0^T \de t P_t(\partial_t \phi_t)=Q(\overline
  \nabla^4 \phi).
\end{equation}

  \begin{definition}[Measure-flux solutions to the
    homogeneous Boltzmann equation]
    \label{def:mfboltz}

    We say that a measure-flux pair $(P,Q)\in \SB$
    is a  solution to the homogeneous Boltzmann equation
    if and only if $Q= Q^{P\otimes P}$, where
    $$Q^{P\otimes P}(\de t,\de v, \de v_*,\de v',\de v'_*)
    \coloneqq \de t \, \frac 12 {B(\de v, \de v_*,\de v',\de v'_*)} 
    P_t(\de v) P_t(\de v_*).$$
  \end{definition}
The above definition is justified by the fact that
  if  $(P,Q^{P\otimes P})\in \SB$
  then
  $P\in C([0,T],\mc P(\mc X))$ solves \eqref{eq:kolm}.

\subsection{Hard spheres}

We first recall the variational formulation for the homogeneous Boltzmann
equation for the  hard-spheres model, 
stated in \cite{BBO}, in which
we prove the equivalence of the weak  homogeneous
Boltzmann equation for $P$ and an entropy inequality for $(P,Q)$.
Here $\mc X=\R^d$ and 
the collision kernel is 
\begin{equation}
  \label{eq:BHS}
B(v,v_*,\de v',\de v'_*)= \frac 12 \int_{S^n} \de n
\, |(v-v_*)\cdot n| \, \delta_{v-((v-v_*)\cdot n)n }(\de v')
\, \delta_{v_*+((v-v_*)\cdot n)n }(\de v'_*)
\end{equation}
where $S^d=\{n\in \R^d: |n| = 1\}$.

Fix $e>0$, 
and define $\mc P_e(\R^d)$ as the set of the probability measure $P$
with $P(v^2/2) \le e$ and $P(v) = 0$.
Denoting by $M_e$ the Maxwellian  of energy $e$ and momentum $0$, 
consider the functional
$$\HH_e(P|M_e) =
\begin{cases}
  \int \de P \log \frac {\de P}{\de v} +
  \frac d2 \left(\log \frac {4\pi e}d + 1 \right) & \text{if } P(v^2/2)\le e\\
  +\infty & \text{otherwise.}
\end{cases}
$$
This is the large deviation  rate function of the empirical
  measure %$\frac 1N \sum_{i=1}^N \delta_{v_i}(\de v)$
  of $N$ velocities
$v_1\dots v_N$
distributed according to  the Haar measure on the surface of $\R^{Nd}$
with fixed momentum and energy, namely 
$\frac 1N \sum_{i=1}^N v_i = 0$, and  $\frac 1N \sum_{i=1}^N v_i^2/2 = e$
(see \cite{KR, BBBC}).

Let $\SB_e$  be  the subset of the pairs $(P,Q)\in \SB$
with $P\in C([0,T],\mc P_e(\R^d))$. {We remark that
  if $(P,Q)\in \SB$ then $Q^{P\otimes P}$ is a finite measure,
  and then measure-flux solutions and weak solutions in the sense
  of \eqref{eq:kac} are the same.}
Denote
by $\Upsilon: [0,T] \times \mc X^2 \times \mc X^2 \to
 [0,T] \times \mc X^2 \times \mc X^2$ the map which
exchanges the incoming and the outgoing velocities:
$\Upsilon (t,v,v_*,v',v'_*)= (t,v',v'_*,v,v_*)$.
We state the following proposition.

\begin{proposition}[Entropy balance for the hard-sphere model]

  \label{lemma:rev}

  $~$

For each $(P,Q)\in \SB_e$
\begin{equation}
  \label{eq:rever}
  \HH_e(P_T|M_e)
  + \E_{}(Q|\Upsilon_\# Q^{P\otimes P}) =
  \HH_e(P_0|M_e)
  +  \E_{}(Q|Q^{P\otimes P})
\end{equation}
\end{proposition}
The proof is in  {\cite[Proposition 3.1]{BBBH}}.

\begin{proposition}[Variational solution to the  homogeneous Boltzmann equation]
  A pair 
  $(P,Q) \in \SB_e$ with $P_0(v^2/2)=e$ and
  $\HH_e(P_0|M_e) < +\infty$ 
  is a measure-flux   solution to the homogeneous Boltzmann equation 
  if and only if  
  $$\HH_e(P_T|M_e)
  + \E(Q,\Upsilon_\# Q^{P\otimes P}) \le 
  \HH_e(P_0|M_e),$$
  or, equivalently,
  $$\HH_e(P_T|M_e)
  + \E(Q|Q^{P\otimes P}) 
  + \E(Q|\Upsilon_\# Q^{P\otimes P}) \le 
  \HH_e(P_0|M_e).$$
\end{proposition}
The proof follows from Proposition \ref{lemma:rev}.

We remark that also in this case the sum  
$\E(Q|Q^{P\otimes P}) 
+ \E(Q|\Upsilon_\# Q^{P\otimes P})$
can be written in terms of
a Dirichlet form and $\E(Q|R^{P\otimes P})$
of a suitable measure $R^{P\otimes P}$, see \cite{BBBO}.

\subsection{Non-reversible collision kernels}

We now extend this formulation to the case of non-reversible microscopic
dynamics.

An equilibrium solution to the  homogeneous Boltzmann equation
$\pi\in \mc P(\mc X)$ satisfies 
\begin{equation}
  \label{eq:bolteq}
  \pi(v) \int_{v_*,v',v'_*\in \mc X}  B(v,v_*,\de v',\de v'_*)
  \pi(\de v_*)  = 
  \int_{v_*,v',v'_*\in \mc X}  B(v',v'_*,\de v,\de v_*)
  \pi(\de v')\pi(\de v'_*).
\end{equation}
We are going to assume a stronger condition, namely 
\begin{equation}
  \label{eq:cond}
  \pi(\de v) \pi(\de v_*) \int_{v',v'_*,\in \mc X}  B(v,v_*,\de v',\de v'_*)
    = \int_{v',v'_*\in \mc X}  B(v',v'_*,\de v,\de v_*)
  \pi(\de v')\pi(\de v'_*).
\end{equation}
Note that the above  equality  implies \eqref{eq:bolteq}. On the other hand
it is weaker than the detailed balance type condition
$$B(v',v'_*,\de v,\de v_*) \pi (\de v') \pi(\de v'_*) = 
B(v,v_*,\de v',\de v'_*) \pi (\de v) \pi(\de v_*),$$
which holds for collision kernels invariant in the exchange of
incoming and outgoing velocities.

We refer to \eqref{eq:cond} as {\bf two-particle factorization} condition.
In fact, consider the 
continuous time Markov chain on $\mc X^2$ with transition kernel $B$.
An equilibrium probability measure  $\alpha^{(2)}\in \mc P ( \mc X^2)$
for this process 
verifies
$$\alpha^{(2)}(\de v,\de v_*) \int_{v',v'_* \in \mc X}  B(v,v_*,\de v',\de v'_*)
= \int_{v',v'_*\in \mc X}  B(v',v'_*,\de v,\de v_*)
\alpha^{(2)}(\de v',\de v'_*).
$$
Condition \eqref{eq:cond} says that $\alpha^{(2)} = \pi\otimes \pi$
is a factorized stationary
probability
measure for the dynamics.

In order to simplify the proofs, we only consider  bounded transition kernels,
in the sense specified below.

\begin{assumption}
   \label{ass:B}

   $~$

 \item[(i)] For any $(v,v_*,v',v'_*)\in \mc X^4$
   $$B(v,v_*,\de v',\de v'_*) = B(v_*,v,\de v',\de v'_*) =
   B(v,v_*,\de v'_*,\de v').$$
   
 \item[(ii)] There exists a probability measure $\pi$ which satisfies
   \eqref{eq:cond} and gives positive measure to any
   open subset of $\mc X$.
   
 \item[(iii)] There exists a bounded Borel function $b:\mc X^4\to [0,+\infty)$
   such that
   $$B(v,v_*,\de v',\de v_*')  = b(v,v_*,v',v'_*) \pi (\de v') \pi(\de v'_*).$$ 
   
 \end{assumption}
 \noindent 
 Observe that example $K_1$, $O_1$, $O_2$ in Section \ref{sez:esempi}
 satisfy Assumption \ref{ass:B}.

 Recall {the typical flux}
 $$Q^{P\otimes P} = \de t \, \frac 12 B(v,v_*,\de v',\de v'_*) P_t(\de v)
 P_t(\de v_*),
 $$
 and set $q^{P\otimes P}_t$ its time density, namely
 $q^{P\otimes P}_t=\frac 12 B(v,v_*,\de v',\de v'_*) P_t(v) P_t(v_*)$.
We denote by $\hat B$ the collision rate of the Boltzmann
time-reversed dynamics
\begin{equation}
\label{def:hatB}
\hat B(v',v'_*,\de v,\de v_*) \pi (\de v') \pi(\de v'_*) = 
B(v,v_*,\de v',\de v'_*) \pi (\de v) \pi(\de v_*).
\end{equation}
i.e.
$$ \hat B(v,v_*,\de v,\de v_*) = b(v',v'_*,v,v_*) \pi(\de v)
\pi(\de v_*).
$$
We set
\begin{equation}
  \label{def:hatQPP}
  \hat Q^{P\otimes P} = \de t \, \frac 12 \hat B(v,v_*,v',v'_*) P_t(v) P_t(v_*),
\end{equation}
and we denote by $\hat q^{P\otimes P}_t$ its time density.
%$\frac 12 \hat B(v,v_*,v',v'_*) P_t(v) P_t(v_*)$.

\begin{proposition}
 Under Assumption \ref{ass:B}, 
 for any $(P,Q)\in \SB$
  \begin{equation}
    \label{eq:cr}
    \Ent(P_T|\pi) + \E(Q|\Upsilon_\# \hat Q^{P\otimes P})
    =  \Ent(P_0|\pi) + 
    \E(Q|Q^{P\otimes P}). 
  \end{equation}
  We intend the above equation in the sense that 
  if one of the two sides of the equality is finite, the other if finite
  and they are equal.
\end{proposition}
\begin{proof}
  The proof is very similar to that of the linear case,
  stated in Proposition \ref{eq:EBI} and proved in Appendix \ref{app:prova}.
  Here we only highlight the main points.

  We assume that the right-hand side of \eqref{eq:cr} is finite. This implies
  that $P_t(\de v) = f(t,v)\pi(\de v)$ for some summable density $f$.
  We regularize the measure-flux pair {in velocity and time, using as
    kernel $g_\eps$ and $\imath_\eta$, respectively.
  The regularized pair $(P^{\eps,\eta},Q^{\eps,\eta})$ is in $\SB$
  and $P_0^{\eps,\eta} = P_0^{\eps}$ for any $\eps,\eta>0$.}
  Moreover, $\de P^{\eps \eta} = f^{\eps,\eta}\de \pi$
  and $\de Q^{\eps,\eta}(\de t,\de v,\de v_*,\de v',\de v'_*) =
  q^{\eps,\eta} \de t \, \pi^{\otimes 4}$, where $f^{\eps,\eta}$ and
  $q^{\eps,\eta}$ are positive, {bounded, {smooth} and  bounded away from zero.}
 In particular, $f^{\eps,\eta}$ is differentiable in time.

  By the balance equation
  \begin{equation*}
  \opde{t} \Ent (P_t^{\eps,\eta} | \pi)
  = \int \de \pi^{\otimes 4} q(t,v,v_*,v',v'_*)
  \log \frac {f_t(v') f_t(v'_*)}
  {f_t(v) f_t(v_*)}
\end{equation*}
Fix $\gamma>0$. We rewrite the above expression as 
\begin{equation}
  \label{eq:bilent}
  \begin{split}
    &\opde{t} \Ent (P^{\eps,\eta}_t | \pi)\\
    &=
        \int \de \pi^{\otimes 4} 
        q^{\eps,\eta}(t,v,v_*,v',v'_*)
        \log \frac {2q^{\eps,\eta}(t,v,v_*,v',v'_*)}
      {b_\gamma(v,v_*,v',v'_*)f^{\eps,\eta}(t,v) f^{\eps,\eta}(t,v_*)}
       \\
     &-  \int \de \pi^{\otimes 4}
       q^{\eps,\eta}(t,v,v_*,v',v'_*)
       \log \frac {2 q^{\eps,\eta}(t,v,v_*,v',v'_*)}
      {b_\gamma(v,v_*,v',v'_*) f^{\eps,\eta}(t,v') f^{\eps,\eta}(t,v_*')},
    \end{split}
  \end{equation}
  where $b_\gamma = b + \gamma$.
  Moreover, by condition \eqref{eq:cond}
  \begin{equation*}
    \int\! \de \pi^{\otimes 4}b_\gamma(v,v_*,v',v'_*)  f^{\eps,\eta}(t,v)
    f^{\eps,\eta}(t,v_*)
    =
    \int\! \de \pi^{\otimes 4}b_\gamma(v,v_*,v',v'_*)  f^{\eps,\eta}(t,v')
    f^{\eps,\eta}(t,v_*').
  \end{equation*}
  By adding and subtracting these two terms in the right-hand side of
  \eqref{eq:bilent}, and 
  by integrating
  in time,
  we obtain the chain rule for the entropy for the regularized 
  measure-flux pair
  \begin{equation}\label{11}
    \Ent (P^{\eps,\eta}_T  | \pi) +
  \E(Q^{\eps,\eta}|\Upsilon_\#
  {\hat Q_\gamma}^{P^{\eps,\eta}\otimes P^{\eps,\eta}}) =
  \Ent (P^{\eps,\eta}_0  | \pi) +
  \E(Q^{\eps,\eta}|Q_\gamma^{P^{\eps,\eta}\otimes P^{\eps,\eta}})
  \end{equation}
  where {$Q_\gamma^P$ and $\hat Q_\gamma^P$ are the typical fluxes with $b$ replaced by $b+\gamma$.}

  As $\eps \to 0$, $\eta\to 0$, $\gamma \to 0$,
  the right-hand side of \eqref{11} converges to 
  $$\Ent (P_0  | \pi) + \E(Q|Q^{P\otimes P}).$$
  {The proof follows the same argument of Appendix \ref{app:prova},
    which uses the  convexity and lower semi-continuity of the entropy.
  A main difference is in the removal of the regularization in time,
  for which the convexity argument cannot be directly applied,
  since $f^{\eps,\eta}(t,v) f^{\eps,\eta}(t,v_*)$ is not the
  time regularization of $f^{\eps}(t,v) f^{\eps}(t,v_*)$.
  This difficulty is circumvented as in the proof of
  \cite[Theorem 5.6]{BBB}, Step 2. We omit the details.}

  By lower semi-continuity,
  we obtain 
  $$\Ent (P_T  | \pi) +
  \E(Q|\Upsilon_\#
  {\hat Q}^{P\otimes P}) \le 
  \Ent (P_0  | \pi) +
  \E(Q|Q^{P\otimes P}).
  $$
  To conclude the proof, we show that the reverse inequality holds,
  by repeating the argument for the reversed trajectory. In this case we
  use the time regularization
that preserves $P_T$ instead of $P_0$.
\end{proof}

Now we can state the following Proposition, under Assumption \ref{ass:B}.

\begin{proposition}
  [Variational solution to the  homogeneous Boltzmann equation]
  \label{def:varsol-BC}
  Let $(P,Q)\in \SB$ be such that $\Ent(P_0 |\pi) < +\infty$.
  Then $(P,Q)$ is a measure-flux
  solution to the  homogeneous Boltzmann
  equation if and only if
  \begin{equation}\label{eq:Bfirst}
    \Ent(P_T|\pi) + \E(Q|\Upsilon_\# \hat Q^{P\otimes P}) 
    \le \Ent(P_0|\pi),
  \end{equation}
  or, equivalently,
  \begin{equation}
    \label{eq:Bsecond}
    \Ent(P_T|\pi) + \E(Q|\Upsilon_\# \hat Q^{P\otimes P}) +
    \E(Q|Q^{P\otimes P}) \le \Ent(P_0|\pi). 
  \end{equation}
  Moreover, 
\begin{equation}\label{eq:Bdecom}
  \E(Q|\Upsilon_\# \hat Q^{P\otimes P}) +  \E(Q|Q^{P\otimes P})
  = 2 E(Q|R^{P\otimes P}) + \int_0^T \de t  \mathrm D^2(P_t)
  \end{equation}
  where, in terms of $f_t = \frac {\de P_t}{\de \pi}$,   
  \begin{equation*}\de R^{P\otimes P} \coloneqq 
  \de t \, \frac 12 {B(v,v_*,\de v',\de v'_*)}\pi(\de v) \pi(\de v_*)
  \sqrt{f_t(v) f_t(v_*) f_t(v') f_t(v'_*)}
  \end{equation*}
  and the functional $ \mathrm D^2(P)$ is
  $$\mathrm D^2(P) \coloneqq \int 
  \frac 12 {B(v,v_*,\de v',\de v'_*)} \pi(\de v) \pi(\de v_*)
  (\sqrt{f(v) f(v_*)} - \sqrt {f(v') f(v'_*)})^2
  $$
\end{proposition}

\begin{proof}
  The two variational inequalities \eqref{eq:Bfirst}, \eqref{eq:Bsecond} 
  follow from Proposition \ref{eq:cr}. The equality \eqref{eq:Bdecom} easily follows from direct computation.
\end{proof}

\subsection{Microscopic interpretation}
\label{sez:bbgky}

The
underlying microscopic dynamics of the homogeneous
Boltzmann equation is a Kac-type walk, i.e. the
continuous time Markov chain on $\mc X^N$ with
 generator 
$$\mc L_N = \frac 1N \sum_{i<j} L_{ij}.$$
Here $L_{ij}$ acts of $F:\mc X^N\to \R$ as 
\begin{equation*}\begin{split}
    L_{ij} F(v_1, \dots v_N) = \int_{(v_i',v_j')\in \mc X^2}
B(v_i,v_j,\de v_i',\de v_j')
\big[F(v_1 \dots v_i \dots v_j \dots v_N)\\
  -
F(v_1 \dots v_i' \dots v_j' \dots v_N)\big].
\end{split}\end{equation*}
An invariant  measure $\alpha^N \in \mc P(\mc X^N)$
satisfies $\mathcal L_N^* \alpha^N = 0$, i.e.  
\begin{equation*}\begin{split}
\alpha^N(\de v_1,\dots \de v_N) &\sum_{i<j} \int_{(v',v'_*)\in \mc X^2}
   { B(v_i,v_j,\de v',\de v'_*)}
    \\&- \sum_{i<j} \int_{(v',v'_*)\in \mc X^2}  B(v',v'_*,\de v_i,\de v_j) 
\alpha^N(\de v_1 \dots \de v' \dots \de v'_* \dots \de v_N) = 0.
\end{split}\end{equation*}
We now show that if $\alpha^N = \nu_N^{\otimes N}$ for some $\nu_N\in \mc P(\mc X)$,
then condition \eqref{eq:cond} is fulfilled.
Since $\alpha^N$ is necessarily exchangeable,
indicating by $\alpha^{N,(k)}$ its
$k$-particles marginal, 
by integrating in $\de v_{k+1}, \dots \de v_N$
we get the following BBJKY hierarchy
of equilibrium equations:
$$
\sum_{i<j\le k } L^*_{ij} \alpha^{N,(k)} + (N-k)C^{k,k+1} \alpha^{N,(k+1)} = 0,
$$
where
\begin{equation}
  \label{eq:ckkp}
\begin{split}
  C^{k,k+1} \alpha^{N,(k+1)}(\de v_1,\dots &\de v_k)
  \\=\sum_{i=1}^k
  \int_{(v_{k+1},v_i',v_{k+1}')\in \mc X^2}
    &\left( B(v_i,v_{k+1},\de v_i',\de v_{k+1}')
      \alpha^{N,(k+1)}(\de v_1,\dots \de v_i, \dots \de v_{k+1}) \right.
  \\
  &\left.-B(v_i',v_{k+1}',\de v_i,\de v_{k+1}) \alpha^{N,(k+1)}
                              (\de v_1,\dots,\de v_i',
  \dots \de v_{k+1}')
\right).
\end{split}
\end{equation}
If $\alpha^N = \nu_N^{\otimes N}$ for some $\nu_N\in \mc P(\mc X)$,
then $C^{1,2} \nu_N^{\otimes 2}=0$, i.e. $\nu_N$ is an equilibrium for
homogeneous Boltzmann equation \eqref{eq:bolteq} which we can indicate
by $\pi$.  As a consequence, $C^{k,k+1} \pi^{\otimes (k+1)} = 0$ for
any $k\le N-1$.  Then, by \eqref{eq:ckkp},
$\sum_{i\le j \le k} L_{ij} \pi^{\otimes k} = 0$. In particular, for
$k=2$, we get the condition \eqref{eq:cond}.

It is easy to check that 
if $\alpha^N = \pi^{\otimes N}$,
the collision kernel of the time-reversed  microscopic process is
the one defined in \eqref{def:hatB}.
Moreover,  by adapting the argument in \cite{BBBO}, 
one can prove a large deviation principle for measure-flux
pairs $(P,Q)$ of both the microscopic process and its time-reversed, 
with dynamical rate functions given by
$\E(Q|Q^{P\otimes P})$ and
$\E(Q|\hat Q^{P\otimes P})$, respectively.

\begin{remark}
  For a general transition kernel $B$, the invariant measure
  $\alpha^N$ is not factorized. Nevertheless, in some case, for instance
  for the Kac's walk for hard spheres,
  where $\alpha^N$ is the Haar measure on
  the surface of fixed specific energy and momentum,
  condition \eqref{eq:cond} is verified.
\end{remark}

\begin{remark}
  In \cite{BBC}
  we show that the convergence of 
  $\alpha^{N,(k)}\to \pi^{\otimes k}$ for any $k$,
  i.e. the ``chaoticity'' of the equilibrium measure,
  is not sufficient 
  to obtain  \eqref{eq:cond}.
  Moreover, we show some example
  in which a variational formulation holds
  in terms of the rate function of the LDP 
  for the empirical measure associated to
  the variables $v_1\dots v_N$, distributed
  with $\alpha^N$. This functional can be different from $\Ent(P_0|\pi)$,
  and the corresponding reverse dynamics can be a non-quadratic Boltzmann
  type equation.
  
  Such possibility does not exclude that also the entropy is a decreasing
  function. The two-particle factorization is a sufficient
  condition for this property.

\end{remark}

\subsection{How to construct collision kernels}

In order to construct collision kernels
of homogeneous Boltzmann equations which satisfy the two-particle factorization
condition, we consider a Markov jump process
on $\mc X^2$.
Let $P(\de v',\de v'_*|v,v_*)$ e 
the  transition probability from $v,v_*$ to $v',v'_*$.
Under suitable assumption, there exists a unique invariant
measure $g(\de v,\de v_*)$, namely 
  \begin{equation}
    \label{eq:g}
    g(\de v,\de v_*) = \int_{\mc X^2} P(\de v,\de v_*|v',v'_*)g(\de v',\de v'_*).
\end{equation}
Choose a probability measure $\pi$ such that $g(\de v, \de v_*)\ll
\pi(\de v) \pi(\de v_*)$, fix $c>0$,  and denote by  $\lambda$  the density of
$cg$  with respect to $\pi\otimes \pi$, namely
$cg(\de v,\de v_*)= \lambda(v,v_*) \pi(\de v)\pi(\de v_*)$.

\begin{proposition}

  \label{prop:strategia}

  $~$
  
  Let $B$ be the collision kernel
  $$B(v,v_*,\de v', \de v'_*) = \lambda (v,v_*)
  P(\de v',\de v'_*|v,v_*).$$
  Then $B$ satisfies the two-particle factorization condition
  \begin{equation}
    \label{eq:condcont}
    \pi(\de v) \pi (\de v_*) \lambda(v,v_*) =
    \int_{(v', v'_*) \in \mc X^2} B(v',v'_*, \de v, \de v_*) \pi(\de v') \pi(\de v'_*).
  \end{equation}
\end{proposition}
The proof easily follows from the definition of $\lambda$ and $B$.
We remark
that  $\pi$ is the equilibrium of the
corresponding homogeneous Boltzmann equation.

\section{Some example}
\label{sez:esempi}

The most studied model in kinetic theory
with non-reversible collision rate
is the  Boltzmann equation for granular media
(see \cite{EP97} and the unpublished reference therein of the same authors, and 
\cite{carrillo2021} for a recent survey).
In this model, the particles interact by means inelastic collisions, in
which the energy is dissipated, while the momentum is conserved.

In the homogeneous case,
the asymptotic state is the $\delta$ in the mean velocity,
and the ``kinetic'' entropy $\int \de v \, \tfrac {\de P}{\de v}
\log  \tfrac {\de P}{\de v}$ is diverging; moreover
the relative
entropy w.r.t. to the equilibrium is not finite
for all interesting initial datum (see e.g. the one dimensional case,
treated in \cite{1997_BCP2,toscani00,2001_BP}).
In this case our variational formulation is inapplicable.

On the other hand, in various application of the kinetic theory (e.g. in economy and sociology), 
many interesting models are based on inelastic
interactions with corrections that return energy to the system,
so that there are non-singular asymptotic states (see \cite{toscani2,toscani}).
The major difficulties in handling  these systems
is the lack of an explicit expression of the 
equilibrium, and it is not known if an H-theorem holds
(see also \cite{vulpiani}).

For models which verifies the two-particle factorization condition 
 \eqref{eq:cond}, 
the H-theorem is a simple consequence of 
the variational formulation.
Unfortunately, this condition is not commonly met.
In this section we construct four examples of continuous system
which satisfy the two-particle factorization.
For three of them, assumption \ref{ass:B} are satisfied, 
and then they enjoy the variational formulation and
the consequent H-Theorem.
For the second Kuramoto type example, the analysis is only formal,
and rigorous results require some extra work,
beyond the aims of this paper.

\subsection{Two  Kuramoto  type models}
\label{sez:kuramoto}

The Kuramoto model is a model for phase variables in $\bb S$
which synchronize, due to an attractive pair interaction \cite{Kuramoto}.
Here we consider a transition kernel inspired to this model.

We fix some notation.
For any pair $(\vt,\vt_*)\in \bb S^2$ there exist uniquely 
$\bar {\vt}\in \bb S$ and $\xi \in (-\pi,\pi)$
such that
$$\left\{
  \begin{aligned}
    &\vt = \bar {\vt} + \xi/ 2\\
    &\vt_* = \bar {\vt} - \xi/ 2\\
  \end{aligned}
\right.
$$
Moreover, $\de \vt  \de \vt_* = \de \bar {\vt} \de \xi$.

The arc-length distance is defined as $d(\vt,\vt_*) = \min_{k\in \bb Z} |\vt - \vt_* - 2k\pi|$.
Let $I_r(\bar {\vt})$ be
the arc of length $r<\pi$ centered in $\bar {\vt}$, i.e.
$$I_r(\bar {\vt}) = \{ \vt| \, d(\vt,\bar {\vt}) < r/2 \}.$$

We consider a transition probability from the incoming pair $(\vt,\vt_*)$
to the outgoing pair  $(\vt',\vt_*')$ such that $\vt'$, $\vt_*'$ are
independently  uniformly
distributed in an arch centered in $\bar \vt$ of length depending on the distance
$|\xi|=d(\vt,\vt_*)$.
More precisely, fix a function  $a(|\xi|) :[0,\pi]\to [0,\pi]$ and
define
$$P(\de \vt',\de \vt'_*| \vt,\vt_*) =
\frac 1{a(|\xi|)^2}
\id_{\{\vt' \in I_{a(|\xi|)}(\bar {\vt})\}}
\id_{\{\vt'_* \in I_{a(|\xi|)}(\bar {\vt})\}}
\de \vt'
\de \vt'_*
$$
Observe that the one-marginal is given by 
\begin{equation}
  \label{eq:P1}
  P_a(\de \vt'| \vt,\vt_*) =
  \frac 1{a(|\xi|)}
\id_{\{\vt' \in I_{a(|\xi|)}(\bar {\vt})\}} \de \vt'
\end{equation}
where we use the subscript to specify the dependence on the function $a$.

Consider for instance $a(|\xi|) = |\xi|$;
in this case, $\vt'$ and $\vt'_*$ are uniformly distributed in
the arc between $\vt$ and $\vt_*$, then almost
surely $d(\vt',\vt'_*) < d(\vt,\vt_*)$; 
then the  equilibrium state is the
singular measure $\delta_{\tilde \vt}(\de \vt)$
for some $\tilde {\vt}$.
In order to obtain a non-reversible binary collision dynamics with
 a non-singular equilibrium,
 we have to add to the model a dispersion mechanics.

\subsubsection*{Example $K_1$}
 As a first example  we introduce a dispersion mechanism such that 
a pair of particles at small distance are spread in the semicircle centered
in $\bar \vt$.
Fix $\delta \in (0,\pi)$ and $a(r) = \pi \id_{ \{ r < \delta\}} + r \id_{
\{r\ge \delta\}}$,  chose $P(\de \vt',\de \vt'_*| \vt,\vt_*) =
P_a(\de \vt'| \vt,\vt_*)P_a(\de \vt'_*| \vt,\vt_*)$,
with  $P_a$ defined in Eq. \eqref{eq:P1}.
As we prove in the appendix \ref{app:kura},
the Markov chain on $\bb S^2$ with transition probability $P_a$ has a
unique invariant measure $g$ which is absolutely continuous
with respect the uniform measure $\de \vt \de \vt_*$,
namely $g(\de \vt, \de \vt_*) = \lambda(\vt, \vt_*) \de \vt \de \vt_*$,
where 
\begin{equation}
  \label{eq:hk1}
  \lambda(\vt, \vt_*) =
  \begin{cases}
    \frac 23  \left( \frac 1\delta - \frac {\delta^2}{\pi^2} \right)
    - \frac 23  \left( \frac 1{\delta^2} + \frac {2\delta}{\pi^3}\right)
    (|\xi| -\delta) & \text{ if } \xi \in [0,\delta] \\
    \frac 23 \left( \frac 1{|\xi|}
      - \frac {|\xi|^2}{\pi^3} \right) & \text{ otherwise}
  \end{cases}
\end{equation}
Here $|\xi| = d(\vt,\vt_*)$.
Then, by proposition \ref{prop:strategia}, 
the collision kernel
$B(\vt,\vt_*,\vt',\vt_*') = \lambda(\vt,\vt_*) P(\vt',\vt'_*|\vt,\vt_*)$
satisfies the two-particle factorization condition, with the uniform
measure on $\bb S$ as the invariant measure,
and the corresponding kinetic equation
has a variational formulation and enjoys the $H$-theorem.

\subsubsection*{Example $K_2$}
As second example
we define a dispersion mechanism such that concentration fails
with small probability, namely with small probability particles
spread on the semicircle centered in $\bar \vt$.
Set $a_0(|\xi|)=|\xi|$ and $a_\pi(|\xi|)=\pi$, fix $\eps \in (0,1)$, and chose
$$P(\de \vt',\de \vt'_*| \vt,\vt_*) = (1-\eps)
P_{a_0}(\de \vt'| \vt,\vt_*)P_{a_0}(\de \vt'_*| \vt,\vt_*)
+ \eps   P_{a_\pi}(\de \vt'| \vt,\vt_*)P_{a_\pi}(\de \vt'_*| \vt,\vt_*).$$
Also in this case there exists a unique invariant measure $g$,
whose density w.r.t. $\de \vt \de \vt_*$ is given by
\begin{equation}
  \label{eq:hk2}
  \lambda(\vt,\vt') = \frac 1{1+2r} \left( \left(\frac \pi{|\xi|}
    \right)^r - \left(\frac {|\xi|}\pi\right)^{1+r}
\right),
\end{equation}
where $r\in (0,1)$ and $r(r+1) = 2(1-\eps)$.
Details
are postponed   in  the appendix.
By proposition \ref{prop:strategia},
the collision kernel
$B(\vt,\vt_*,\vt',\vt_*') = \lambda(\vt,\vt_*) P(\vt',\vt'_*|\vt,\vt_*)$
satisfies the two-particle factorization condition and the invariant
measure is  the uniform measure on $\bb S$.

\vskip.3cm
In both the example,  particles which are close interact more often.
This can be seen  Fig. \ref{fig:h} where the collision rate $\lambda$
as a function of the distance $|\xi|$,
is plotted. This is necessary for having an effective dispersion mechanism,
which assures that the two-particle factorization condition holds.

\begin{figure}[h]
  \centerline{
    \includegraphics[scale=0.3]{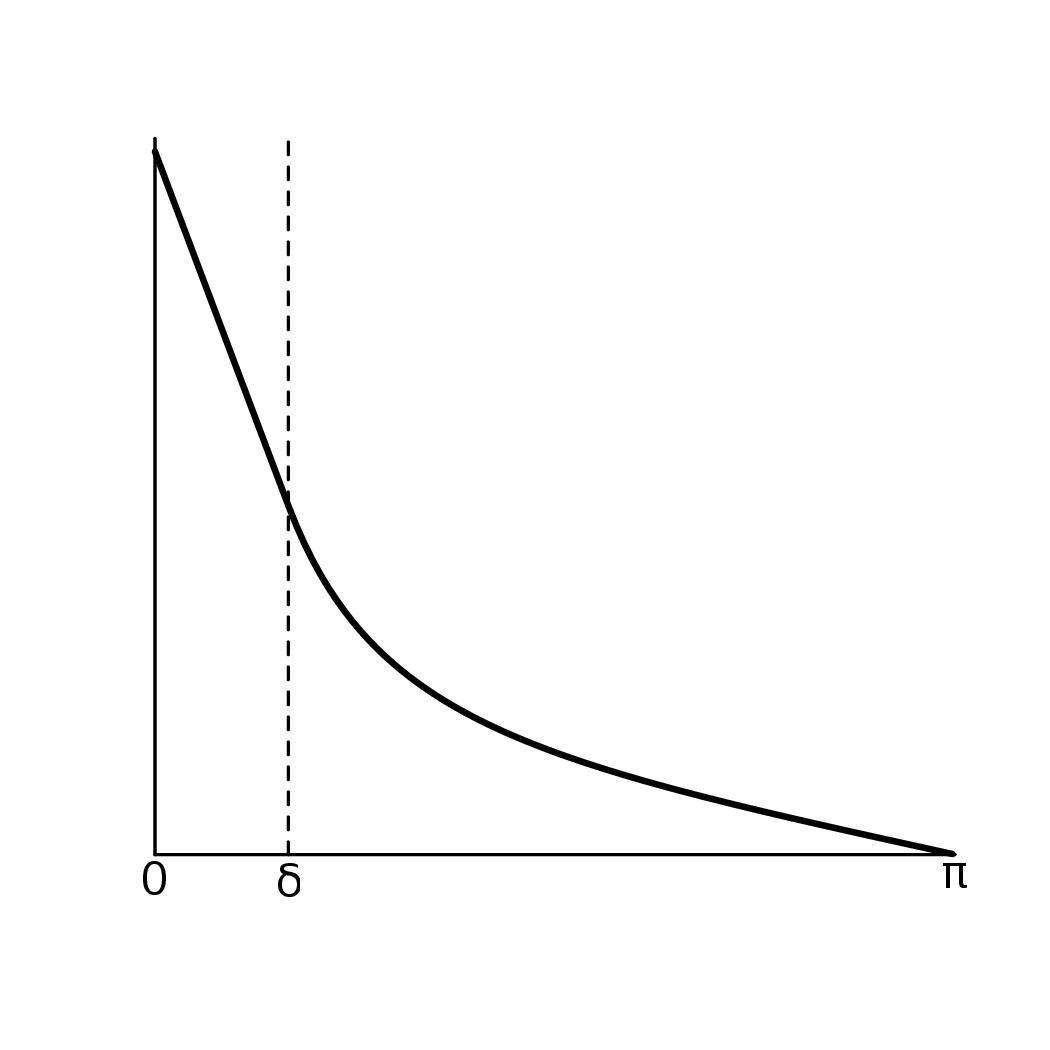}
    \includegraphics[scale=0.3]{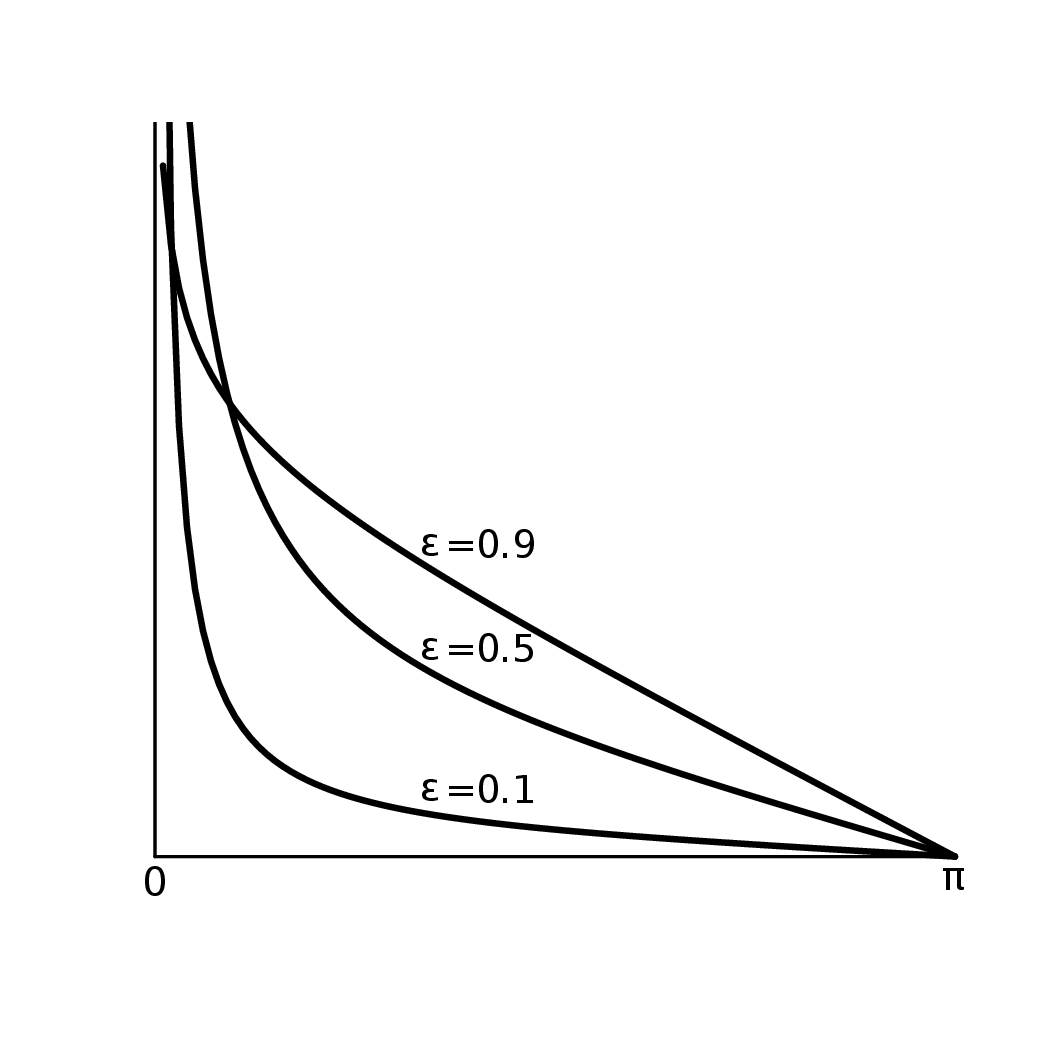}
  }
  \caption{On the left, the collision rate $\lambda$
    for the first example as in Eq. \eqref{eq:hk1} for $\delta = \pi/6$;
    on the right, the collision rate for the second example as in 
    Eq. \eqref{eq:hk2} for $\eps = 0.1, 0.5, 0.9$. }
  \label{fig:h}
\end{figure}

\begin{remark}
  The Kuramoto model for identical oscillators 
  preserves the mean phase.
  This is not true for our model for which ${\bar \vt}'\neq \bar \vt$.
  We can recover this conservation law by using
  transition probabilities of the form
  $$P(\vt',\vt'_*|\vt,\vt_*) = 
  \frac 12 P_a(\de \vt'|\vt,\vt_*)
  \delta_{\bar \vt - \vt'}(\de \vt_*')+
  \frac 12 P_a(\de \vt'_*|\vt,\vt_*)
  \delta_{\bar \vt - \vt_*'}(\de \vt').
$$
\end{remark}

\subsection{Opinion dynamics models}

The second class of models we present is about the opinion formation,
and are inspired of \cite{toscani}.
The opinion is identified with a state $v\in \mc X \coloneqq [-1,1]$;
in a discussion between two agents, their opinions $v,v_*$
change according to
a rule of the following type:
\begin{equation}
  \begin{split}
    &v'= (1-D(|v|)) v + D(|v|) v_* + \xi(v,v_*)\\
    &v_*'= D(|v_*|) v +  (1-D(|v_*|)) v_*   + \xi(v_*, v). 
  \end{split}
\end{equation}
Here $D:[0,1]\to [0, 1/2]$
while $\xi(v,v_*)$
are random variables with zero mean and 
finite variance, chosen such that $v',v'_*\in \mc X$.
If $\xi\equiv 0$ and $D=\eps \in (0,1)$, we recover  the inelastic
collision rule with fixed restitution coefficient $1-2\eps$, 
which preserves the ``mean opinion''
$\frac {v+v_*}2  = \frac {v'+v'_*}2$, but reduces the difference
$|v'-v'_*| = |1-2\eps|\, |v-v_*|$ since $|1-2\eps|<1$.
This inelastic behavior, which holds also if $D$ is non-constant,
takes into account the tendency of interacting people
to bring their opinions
closer together.
The random terms take into account
external factors which can modify the outcome of the interaction.
In a model in which  $v\approx 0$ is a ``weak'' opinion,
and $v \approx \pm 1$ are ``strong'' opinions, it is  assumed
that $D$ is a decreasing function.

\subsubsection*{A symmetric model - example $O_1$}

$~$

We consider $D(|v|) = (1-|v|)/2$, 
and we choose $\xi(v,v_*)$ as a Gaussian variable
with zero mean and variance $\sigma(v)^2 = \delta^2 (1-|v|)^2$,
with $\delta >0$. By  conditioning $\xi(v,v_*)$ to the $v'\in \mc X$, and
$\xi(v_*,v)$ to the $v'_*\in \mc X$, we define the transition probability 
\begin{equation}
  \label{eq:p1}
  \begin{aligned}
  P(\de v',\de v'_*|v,v_*) &=
  \frac 1{ z(v,v_*)}
  M_{\sigma(v)}( v' - ((1-D(v)) v + D(v) v_*))
  \\
  &\times 
  \frac 1{ z(v_*,v) }
  M_{\sigma(v_*)}( v'_* - ((1-D(v_*)) v_* + D(v_*) v'))
 \de v' \de v'_*,  \end{aligned}
\end{equation}
where $M_{\sigma}:\R\to \R$ is the Gaussian
of mean $0$ and
variance $\sigma$, and 
$$z(v,v_*) = \int_0^1 \de v'
M_{\sigma(v)}( v' - ((1-D(v)) v + D(v) v_*)).$$
We have also to symmetrize \eqref{eq:p1}
in order to satisfy item (1) of Assumption \ref{ass:B}.

We numerically find the unique  equilibrium $g$
of the two-particle process, as defined in \eqref{eq:g}. 
Then we look for a strictly positive function $\pi(v)$ and
a strictly positive  function $\lambda(v,v_*)$  decreasing in  $|v-v_*|$,
such that 
$g(\de v,\de v_*) = \lambda(v,v_*)\pi(\de v) \pi(\de v_*)$.
We remark that the decreasing behavior of $\lambda$ 
is what we expect in real relations, where agents with
very different opinions have very few interactions.

In the first graph in  Fig. \ref{fig:mod4}
we represent the values of $g$ obtained  
numerically, with $\delta =1/100$; the values are
increasing from salmon pink to blue.
The small value of $g$ for $(v,v_*) \approx \pm (1,1)$
are due to a boundary effect
related to the discretization.
As a candidate for $\pi$ we have considered a positive power of
the one marginal of $g$.
In the second graph  we choose the power equal to $0.65$.
The resulting $\lambda(v,v_*)$ is shown in the rightmost
graph, and
it turns out to be  approximately
a decreasing function of $|v-v_*|$.
Observe that $\pi$ exhibits
two peaks near the ``extreme'' opinions $\pm 1$,
as showed in the central graph.

\begin{figure}[h]
  \centerline{
    \includegraphics[scale=0.2]{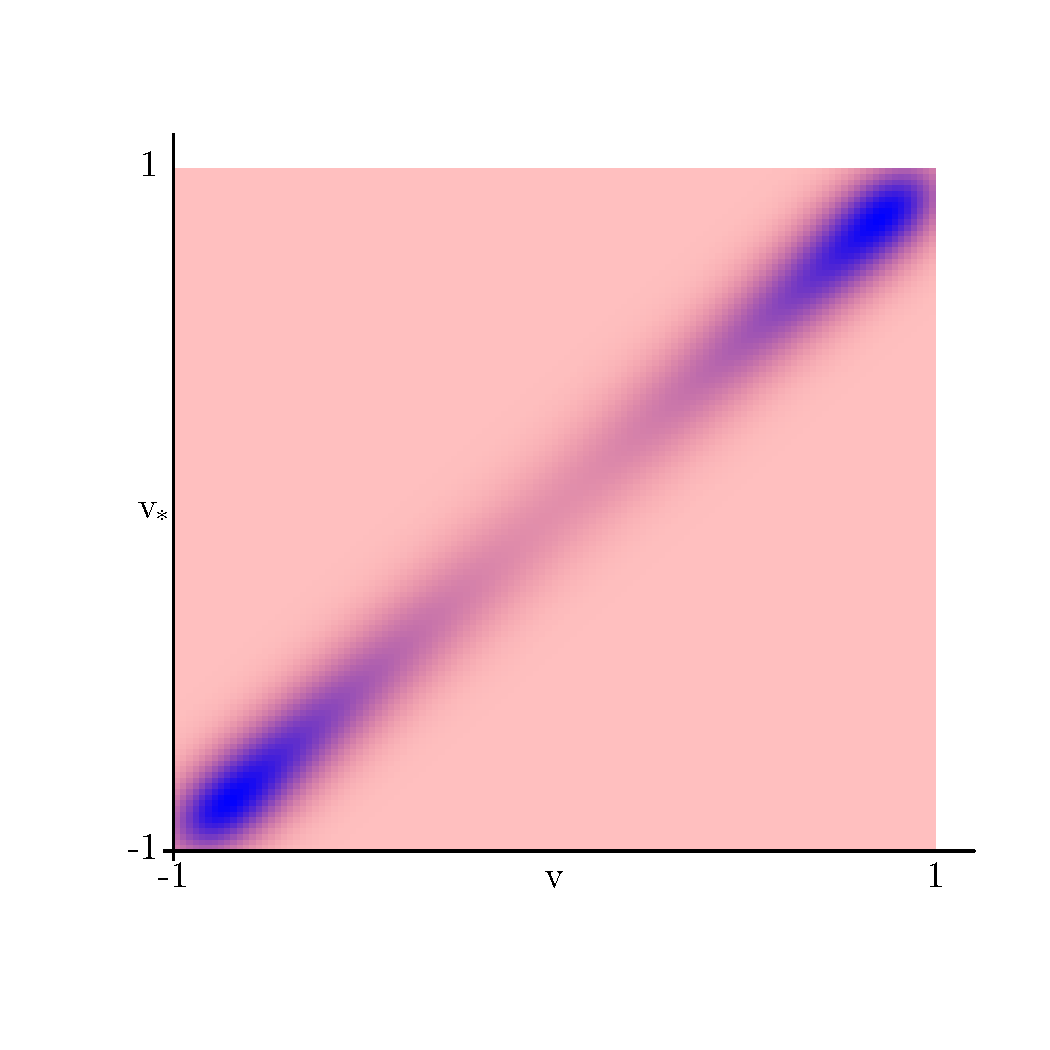}
    \includegraphics[scale=0.2]{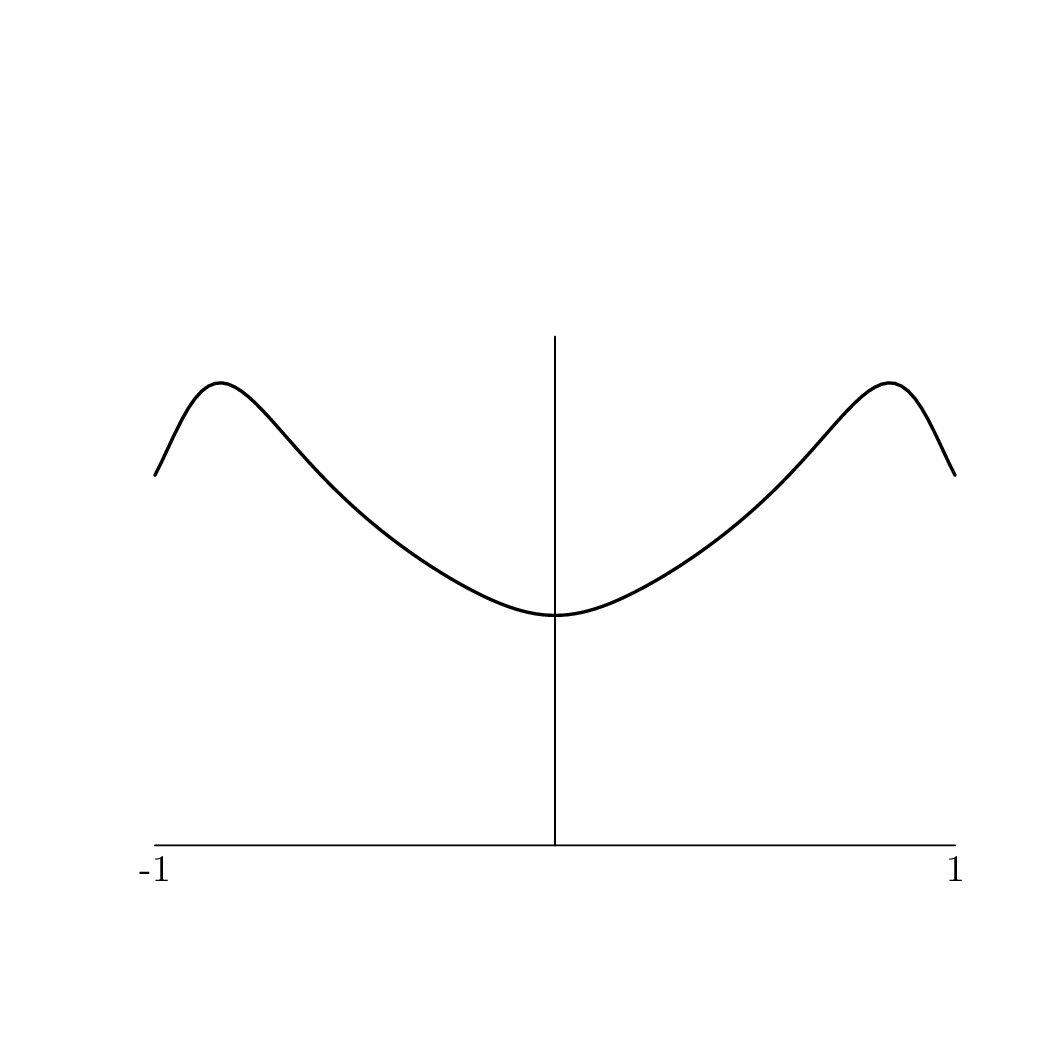}
    \includegraphics[scale=0.2]{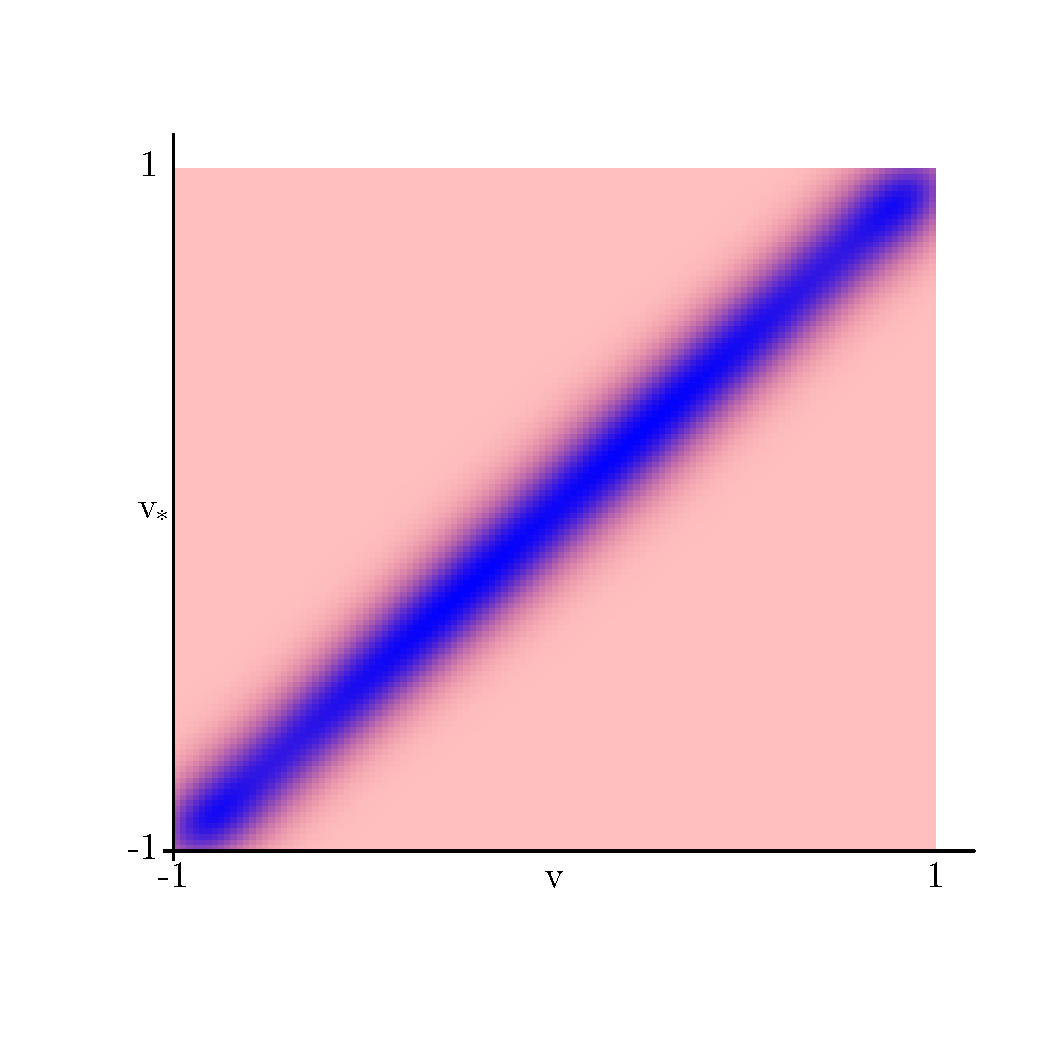}
  }
  \caption{The symmetric model. From the left to the right, the level
    sets of $g$,
     the equilibrium $\pi$ and the level sets of $\lambda$.}
  \label{fig:mod4}
\end{figure}

\subsubsection*{An asymmetric model - example $O_2$}

$~$

We consider a modification of the previous model by introducing a
``repulsion'' mechanism, which strongly increases the transitions from
two agents with opinions close to $1$ to a situation in which one of
the two agents abandons the extreme opinion, i.e.
$(v',v'_*) \approx (1,0)$ or $(0,1)$.  This is obtained multiplying $P$
in \eqref{eq:p1} by a term exponentially small in
$|v-1|^2+|v_*-1|^2+\max(|v'-1|^2+|v'_*|^2, |v'_*-1|^2+|v'|^2)$, and
normalizing.  The results are shown in Fig. \ref{fig:mod5}. At
equilibrium, the opinions $v\approx 1$ flowed back to the neighbor of
0.  It might be interesting to observe the transition from the first
model to this one, slowly varying the parameter that regulates its
relative importance, but this is not the main purpose of this work.

\begin{figure}[h]
  \centerline{
    \includegraphics[scale=0.2]{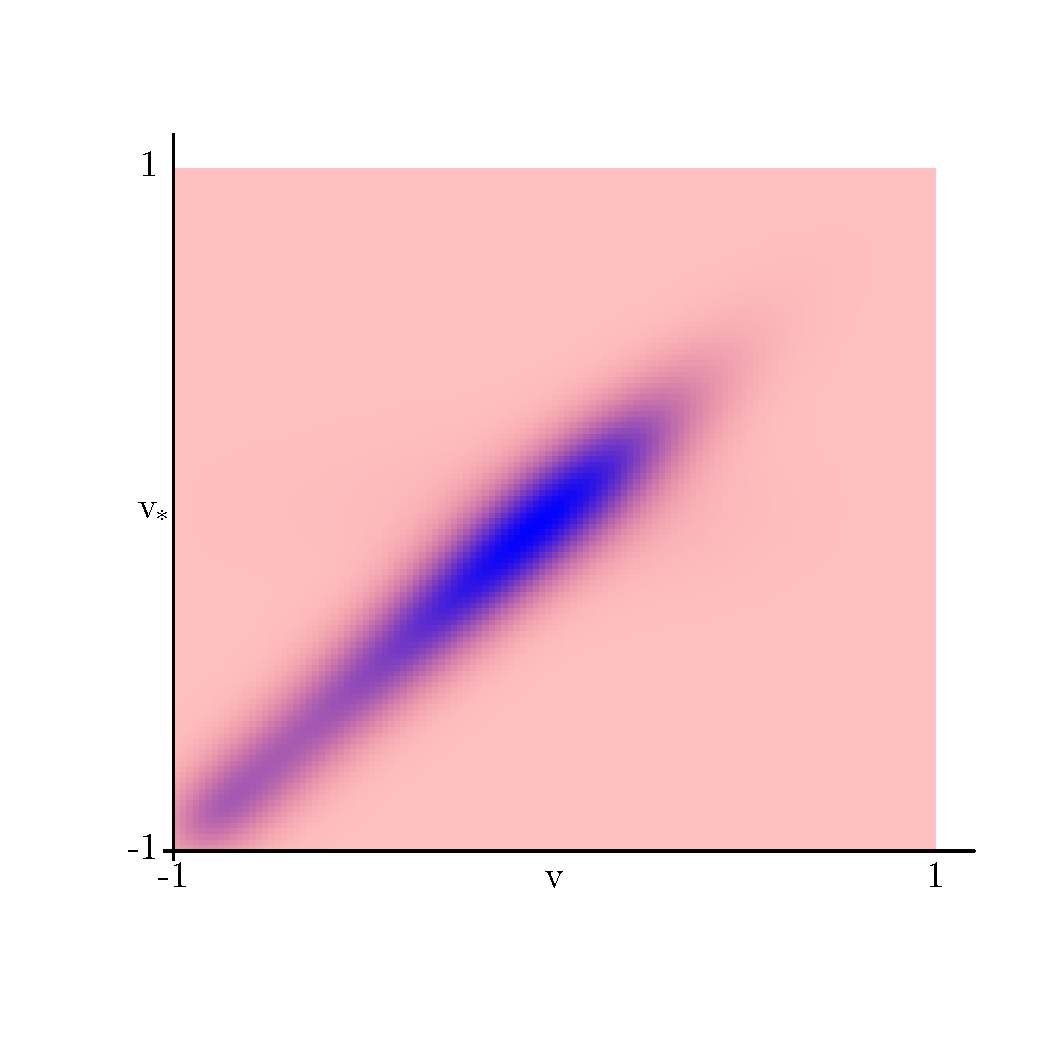}
    \includegraphics[scale=0.2]{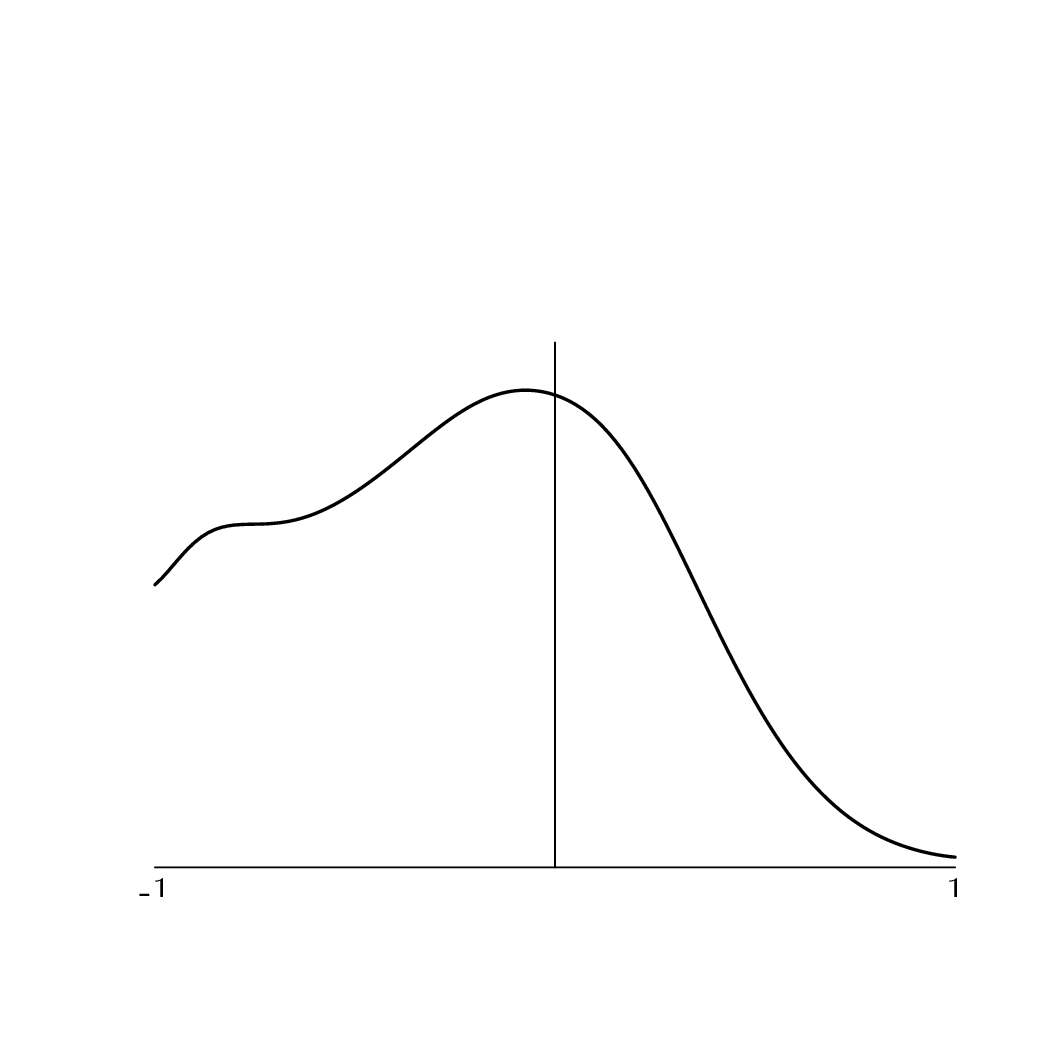}
    \includegraphics[scale=0.2]{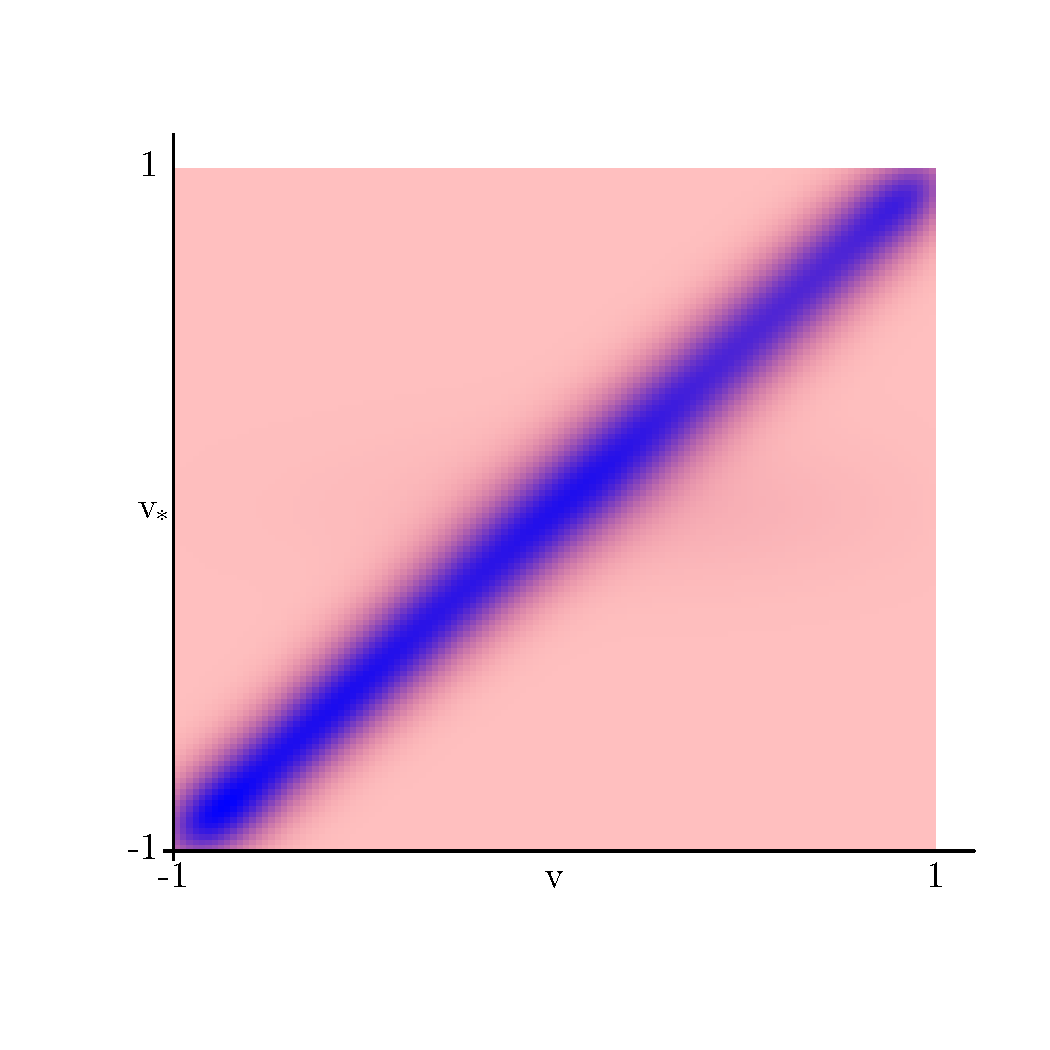}
  }
  \caption{The asymmetric model. %$g$, $\pi$, $\lambda$ in the asymmetric model.
    From the left to the right, the level sets of $g$,
    the equilibrium $\pi$ and the level sets of $\lambda$.
    The 
    repulsive mechanism can be observed in the graph of $g$.
    In the graph of $\pi$ one can observe that 
    neutral opinions are favored while  opinions near $v=1$ are disfavored.
  }
  \label{fig:mod5}
\end{figure}

\appendix

\section{Entropy balance for non-reversible Markov chains
with bounded transition kernel}
\label{app:prova}

In what follows,
it will be useful to explicit the dependence on $\sigma$ of the typical flux. We will write
$$\mc V^P_{\sigma}(\de t,\de x,\de y) \coloneqq \de t P_t(\de x) \pi(\de y)
\sigma(x,y), \ \ 
\Upsilon_\# \hat {\mc V}^P_{\sigma}(\de t,\de x,\de y)
\coloneqq \de t P_t(\de y) \pi(\de x)
\sigma(x,y).$$

\subsection*{Regularizing kernels}

We suppose that there exist a ``regularizing kernel''  i.e. a family of functions $g_\eps\in C_b(\mc X\times \mc X)$, $\eps >0$
such that 
\begin{enumerate}[i)]
\item $g_\eps(x,y)=g_\eps(y,x)\ge \eps$ for all $(x,y)\in \mc X\times \mc X$;
\item $\int_{\mc X} g_\eps(x,y) \pi(\de y)= 1$ for all $x\in \mc X$;
\item for any bounded Borel function $\phi$,
  the continuous function 
  $$G_\eps \phi(x) \coloneqq \int_{\mc X} g_\eps(x,y) \phi(y) \pi(\de y)$$
  converges $\pi$ a.e. to $\phi(x)$; if $\phi\in C_b(\mc X)$ the convergence
  is pointwise.
\end{enumerate}
The details of the construction of such regularizing kernel depends
on $\mc X$ and
on the equilibrium measure $\pi$. As an example, 
consider $\mc X = \R^d$, $\pi(\de x) = h(x) \de x$ with $h>0$ and
regular. Let $\e^{\eps \mathcal L}$ be the semigroup
associated to the differential operator
$\mathcal L\alpha  = \frac 1h \nabla \cdot  ( h \nabla \alpha)$.
Its kernel $k_\eps$ with respect to $\pi$ is symmetric.
Therefore, $g_\eps = (1-\eps) k_\eps + \eps$ is
a regularizing kernel in the sense specified above.

We remark that the result in this appendix can 
be obtained also by using non-symmetric regularizing kernel,
with a few more details to consider.

\vskip.3cm
We  regularize the probability measure by duality.
Given $\mu \in \ms P(\mc X)$, 
$$(G^*_\eps \mu) (\de x) \coloneqq \pi(\de x)
\int_{\mc X} g_\eps(y,x) \mu(\de y)$$
is a probability measure which 
weakly converges to $\mu$ as $\eps \to 0$.

\vskip.3cm

For the  regularizing in time, we consider a non-negative kernel
$\imath_\eta \in C^1(\R)$, $\eta > 0$ with support in $(0,\eta)$,
which approximate the Dirac measure $\delta_0$ in $\bb R$
as $\eta \to 0$.
Given $h\in C((-\infty,T])$,
we define its time regularization as
$$h^\eta(t) \coloneqq \int_{t-\eta}^t \imath_\eta(t-s) \tilde h(s) \de s.$$
Given $h\in C([0,T])$, we define 
$$
\tilde h(t) = \begin{cases}
  h(t) & \text{ if } t\in [0,T]\\
  h(0) & \text{ if } t<0
\end{cases}\ \text{ and }
\check h(t)= \begin{cases}
  h(t) & \text{ if } t\in [0,T]\\
  0 & \text{ if } t<0,
\end{cases}
$$
whose time regularization are indicated by
$\tilde h^\eta$ and $\check h^\eta$ respectively.
Observe that if $h$ is only summable,
the regularization $h^\eta$ only converges to $h$ in $L^1$ and almost
everywhere.

\vskip.3cm

\subsection*{Regularization of the dynamics}

Consider now a pair $(P,\mc V)\in \SM$, and
assume that the  right-hand side of \eqref{eq:EBI} is finite, i.e. 
$$ \Ent (P_0|\pi) + E(\mc V|\mc V^P_\sigma)<+\infty.$$
First we prove that this condition assure
that $P_t\ll \pi$ for any $t\in [0,T]$.
The continuity equation can be extended to characteristic functions
of compact set $\chi_K$, obtaining, for any $t\in [0,T]$
$$\begin{aligned}
  \int \chi_K(x) P_t(\de x) + \int \chi_{[0,t]}(s)
 &\chi_K(x) \mc V(\de s,\de x,\de y) \\
&= \int \chi_K(x) P_0(\de x) + \int
\chi_{[0,t]}(s) \chi_K(y) \mc V(\de s,\de x,\de y).
\end{aligned}$$
If  $\int \chi_K (y)\pi(\de y)= 0$, 
then $\int \chi_K(y) \mc V^P(\de t, \de x, \de y)=0$.
Conditions $\Ent (P_0|\pi) <+\infty$ and $E(\mc V|\mc V^P_\sigma)<+\infty$ imply $P_0\ll \pi$ and $\mc V\ll \mc V^P$, therefore
$\int  \chi_K(x) P_0(\de x) = 0$ and $\int
 \chi_K(y) \mc V(\de t, \de x, \de y)=0$.
 By the continuity equation,  $\int \chi_A(x) P_t(\de x)  = 0$.
This fact implies that $P_t\ll \pi$ for all $t\in [0,T]$.
From now on, we will denote by  $f(t,x)$ the density of $P_t$ with respect to $\pi$.

Define the regularized pair $(P^\eps, \mc V^\eps)$ as
$P^\eps_t = G^*_\eps P_t$, $t\in [0,T]$ 
and 
$\mc V^\eps = (\id \otimes G^*_\eps \otimes G^*_\eps) \mc V$.
It is easy to verify that $(P^\eps,\mc V^\eps)\in \SM$.
The approximated measure is
$P_t^\eps(\de x) = f^\eps(t,x)  \pi(\de x)$,
with $f^\eps(t,x) =  \int_{\mc X} g_\eps(y,x) P_t(\de y)
= G_\eps f(t,x) \in C_b([0,T],\mc X)$ and bounded away from zero.
Since $\E(\mc V|\mc V_\sigma^P)<+\infty$ implies
$\mc V \ll  \mc V_\sigma^P$, then
$\mc V^\eps(\de t, \de x, \de y)
= \de t \, \pi(\de x) \de \pi (\de y) q^\eps(t,x,y)$,
where 
$q^\eps(t,x,y)$ is {strictly} positive, continuous and bounded in $(x,y)$ for
any $t\in [0,T]$.

Finally, we regularize the extension $\tilde f^\eps$ and
$\check q^\eps$, obtaining
the measure 
$P^{\eps,\eta}_t(\de x) = \pi(\de x) \tilde f^{\eps,\eta}(t,x)$ and the flux  
$\mc V^{\eps,\eta}(\de t,\de x,\de y) =  \de t \, \pi(\de x) \de \pi (\de y)
\check q^{\eps,\eta}(t,x,y)$.
It is easy to verify that also 
$(P^{\eps,\eta}, \mc V^{\eps,\eta})$ belongs to $\SM$.
The regularity of $f^{\eps,\eta}$ and  $q^{\eps,\eta}$ allows us to
write the differential
version of the continuity equation:
$$\partial_t f^{\eps,\eta}(t,x) = \int_{\mc X} \pi(\de y)
( q^{\eps,\eta}(t,y,x) - q^{\eps,\eta}(t,x,y)).$$

\vskip.3cm
Now we can state the entropy balance for the regularized dynamics.
Since $ f^{\eps,\eta}$ is bounded, strictly positive and regular
in time, we can compute the derivative of the relative entropy
$$
\begin{aligned}
  \opde{t} \Ent(P^{\eps,\eta}_t|\pi)
  &=
    \int_{\mc X \times \mc X} \pi(\de x) \pi(\de y) q^{\eps,\eta}(t,x,y) \log \frac {f^{\eps,\eta}(t,y)}{f^{\eps,\eta}(t,x)},
\end{aligned}
$$
which can be written as 
\begin{equation}
  \label{eq:derivata}
\begin{aligned}
  \opde{t} \Ent(P^{\eps,\eta}_t|\pi)
 &=\int_{\mc X \times \mc X} \pi(\de x) \pi(\de y) q^{\eps,\eta}(t,x,y) \log
    \frac {q^{\eps,\eta}(t,x,y)}{(\sigma(x,y) + \gamma) f^{\eps,\eta}(t,x)}
  \\
  &-\int_{\mc X \times \mc X} \pi(\de x) \pi(\de y) q^{\eps,\eta}(t,x,y) \log
    \frac {q^{\eps,\eta}(t,x,y)}{(\sigma(x,y) + \gamma) f^{\eps,\eta}(t,y)}
\end{aligned}
\end{equation}
for any positive $\gamma$.
By Integrating in time \eqref{eq:derivata}
and using that 
$\mc V_{\sigma+\gamma}^P(1) = \Upsilon_\# \hat {\mc V}_{\sigma+\gamma}^P(1)$, 
we obtain  
$$\Ent (P^{\eps,\eta}_T|\pi) + E(\mc V^{\eps,\eta}|\Upsilon_\#
{\hat {\mc V}}_{\sigma + \gamma}^{P^{\eps,\eta}}) =
\Ent (P^{\eps,\eta}_0|\pi) + E(\mc V^{\eps,\eta}|\mc V_{\sigma + \gamma}^{P^{\eps,\eta}}),$$
which is the chain rule of the entropy for regularized
measure-flux pairs, with respect to the process with perturbed rate $(\sigma(x,y) + \gamma)\pi(\de y)$.
By the particular choice of
the time regularizing kernel,  $P^{\eps,\eta}_0$ in the right-hand side is equal
to $P^{\eps}_0$
for any $\eps, \eta\ge 0$.

\subsection*{Convergence}

We prove the convergence of the relative entropy.
By convexity
$$
\Ent(P_0^\eps|\pi) \le \Ent(P_0|\pi).$$
Recall that $P_0^\eps\to P_0$ weakly as $\eps \to 0$.
Using the previous inequality and the lower semi-continuity
of the relative entropy  we conclude that 
$$\lim_{\eps \to 0} \Ent(P^\eps_0|\pi) = \Ent(P_0|\pi).$$

We prove the convergence of the flux term.
For any $\gamma> 0$
$$\E(\mc V| \mc V_{\sigma + \gamma}^P) \leq \E(\mc V| \mc V_\sigma^P) + \gamma < +\infty,$$
and
$$\E(\mc V| \mc V_{\sigma + \gamma}^P)
= \E(\mc V| \mc V_1^P) + \mc V\big(\log \frac 1 {\sigma + \gamma}\big) +
\mc V_{\sigma+\gamma }^P(1) - T.$$
Since $|\log (\sigma + \gamma)|$
is bounded, all terms in
the right-hand side are bounded.
Analogously,
$$ \E(\mc V^{\eps,\eta}| \mc V_{\sigma + \gamma}^{P^{\eps,\eta}})  =
\E(\mc V^{\eps,\eta}| \mc V_1^{P^{\eps,\eta}})+
\mc V^{\eps,\eta}\big(\log \frac 1 {\sigma + \gamma}\big)
+\mc V_{\sigma+\gamma}^{P^{\eps,\eta}}(1)   - T$$
By convexity and semi-continuity
$$
\lim_{\eps,\eta \to 0}\E(\mc V^{\eps,\eta}| \mc V_1^{P^{\eps,\eta}}) =  \E(\mc V| \mc V_1^P).$$
$\mc V_{\sigma }^{P^{\eps,\eta}}(1)$ converges to $\mc V_{\sigma}^P$.
Moreover, since $\sigma + \gamma$ is strictly positive and bounded and
  $\mc V\ll \pi\otimes \pi$, 
$$\lim_{\eps,\eta  \to 0} \mc V^{\eps,\eta}(\log (\sigma+\gamma)) =  \mc V(\log (\sigma+\gamma)).$$
Collecting these facts, we obtain
$$\lim_{\eps \to 0}  \E(\mc V^{\eps,\eta}| \mc V_{\sigma + \gamma}^{P^{\eps,\eta}}) =  \E(\mc V| \mc V_{\sigma + \gamma}^P)$$
Finally, by dominated convergence
$\E(\mc V| \mc V_{\sigma+\gamma}^P)$ goes to $\E(\mc V| \mc V_{\sigma}^P)$ as $\gamma \to 0$.

\vskip.3cm

Passing into the limit $\eps\to 0$, $\eta \to 0$,
$\gamma\to 0$ in the chain rule for the regularized dynamics, 
by lower semi-continuity, 
we obtain that also the left-hand side of \eqref{eq:EBI} is finite:
$$\Ent(P_T|\pi) + E(\mc V|\Upsilon_\# \hat {\mc V}^P) \le \Ent(P_0|\pi) + 
E(\mc V|{\mc V}^P)<+\infty.$$

To conclude the proof, we obtain the reverse inequality 
by repeating the argument for the reversed trajectory. In this case we
have to use the time regularization
that preserves $P_T$ instead of $P_0$.

\section{Invariant measure for Kuramoto type model}
\label{app:kura}

The expressions of the collision rates $\lambda$
in \eqref{eq:hk1} and \eqref{eq:hk2},
for the two Kuramoto type models,
are obtained by solving Eq. \eqref{eq:g} for $g$, and setting
$g(\de \vt, \de \vt_*)=\lambda(\vt,\vt_*)\de \vt \de \vt_*$.
In this way the equilibrium distribution is the uniform probability measure on
$\bb S$.

We consider the first model, for which
$$P(\de \vt',\de \vt'_*| \vt,\vt_*) =
P_a(\de \vt'| \vt,\vt_*)P_a(\de \vt'_*| \vt,\vt_*)$$
where  $a(r) = \pi \id_{ r < \delta} + r \id_{r\ge \delta}$
and $P_a$ is defined in \eqref{eq:P1}.
We can express $P$ in terms of $\bar {\vt '}$ and $\xi'$,
where $\vt'= \bar {\vt'} + \xi'/2$, $\vt'_*= \bar {\vt'} - \xi'/2$:
\begin{equation}
\label{eq:P1P1}
  P_a(\de \vt'| \vt,\vt_*)P_a(\de \vt'_*| \vt,\vt_*) 
  = \frac 1{a^2(|\xi)} \id_{ a(|\xi|) > |\xi'|}\,
  \id_{ |\bar {\vt} -\bar{ \vt'}| < a(|\xi|)/2 - |\xi'|/2} \de \bar {\vt '}\de \xi'.
\end{equation}
It is easy to understand that 
$\lambda$ is transitionally
invariant because the dynamics are, then $\lambda(\vt,\vt') = h(\xi)$ which 
satisfies
\begin{equation*}
  % \label{eq:eqm1}
  h(\xi) = \int_{\bb S} \de \bar \vt \int_0^\pi \de \xi'
  \frac 1{a^2(|\xi')} \id_{ a(|\xi'|) > |\xi|}
  \id_{ |\bar {\vt} -\bar{ \vt'}| < a(|\xi'|)/2 - |\xi|/2} h(\xi')
\end{equation*}
By integrating in $\bar \vt$ and observing that $h(\xi) = h(|\xi|)$,
this equation becomes, for $\xi\in [0,\pi]$, 
\begin{equation}
  \label{eq:h1}
  h(\xi) = 2 \int_0^\pi \de \xi' 
  \frac {[a(|\xi'|) - |\xi|]^+}{a^2(|\xi')}  h(\xi')
\end{equation}
It is easy to prove that: $h\in C^1([0,\pi])$,
is affine in $[0,\delta]$, verifies
$h(\pi) = 0$, $h'(\pi) = -2 C$ whit $C= \int_0^\delta h / \pi^2$,
and for $\xi > \delta$
$$h''(\xi) = 2h(\xi) / \xi^2.$$
This second order linear equation is solved by linear combination
of $\xi^{-1}$ and $\xi^2$. Imposing the compatibility condition
in $\xi = \delta$ we obtain 
\eqref{eq:hk1}, where we fixed $C=1$.
We remark that if $\delta \to 0$ the equilibrium $h$ becomes proportional to
$\delta_0(\de \xi)$ as expected, since the dispersion mechanism is removed.

In the second model,
$a_0(r)=r$ and $a_\pi(r)=\pi$, fix $\eps \in (0,1)$, and 
$$P(\de \vt',\de \vt'_*| \vt,\vt_*) = (1-\eps)
P_{a_0}(\de \vt'| \vt,\vt_*)P_{a_0}(\de \vt'_*| \vt,\vt_*)
+ \eps   P_{a_\pi}(\de \vt'| \vt,\vt_*)P_{a_\pi}(\de \vt'_*| \vt,\vt_*).$$
Also in this case $\lambda$ is of the form 
$\lambda(\vt,\vt') = h(\xi)$; now $h$ verifies the equation
$$
h(\xi) = 2 (1-\eps)\int_{\xi}^\pi \de \xi' 
\frac {\xi' - \xi}{{\xi'}^2}  h(\xi')
+ 2\eps (\pi - \xi) K
$$
where $K=\int_0^\pi h/\pi^2$.
Note that $h(\pi) = 0$.
Assuming $h\in C^2(\bb S)$, by deriving in $\xi$ the equation
we obtain that $h'(\pi) = 2\eps \pi$.
By deriving two times we have that
$$h''(\xi) = 2(1-\eps) \frac {h(\xi)}{\xi}.$$
Then $h$ is a linear combination of $\xi^{-r}$ and $\xi^{1+r}$,
where $r\in (0,1)$ and $r(r+1) = 2(1-\eps)$.
By imposing the two conditions in $\xi = \pi$ we get
\eqref{eq:hk2}, where we have fixed $K=1$.
Also in this case when $\eps \to 0$ the equilibrium $h$ becomes a $\delta$ function.

\section*{Acknowledgments}

D. Benedetto  would like to thank GNFM - INdAM.

\section*{Declarations}

The authors have been  supported  by  PRIN
202277WX43 ``Emergence of condensation-like phenomena in interacting
particle systems: kinetic and lattice models'',  founded by the European
Union - Next Generation EU.

The authors have no conflicts of interest to declare.

Data sharing is not applicable to this article.

\end{document}